\pdfoutput=1
\documentclass[12pt, draftclsnofoot, onecolumn]{IEEEtran}

\usepackage{bbm}
\usepackage{amsthm}
\usepackage{amsmath}
\usepackage{amsfonts}
\usepackage{amsmath,amsfonts,amssymb,amsbsy, amsthm,ae,aecompl}
\usepackage{algorithm,algpseudocode}
\usepackage[english]{babel}
\usepackage{bm}
\usepackage{color}
\usepackage[noadjust]{cite}
\usepackage{epsfig}
\usepackage{enumerate}
\usepackage{float} 
\usepackage{fancyhdr}
\usepackage[T1]{fontenc}
\usepackage[acronym,toc,shortcuts]{glossaries}
\usepackage{graphicx, caption, subcaption}
\usepackage{hyperref}
\usepackage{lastpage}
\usepackage{listings}
\usepackage{lipsum}
\usepackage{dsfont}
\usepackage{multirow,tabularx}
\usepackage[normalem]{ulem}
\usepackage{tikz,pgfplots}
\usepackage{times}
\usepackage{verbatim}
\usepackage[all]{xy}
\usepackage{mathtools}
\usepackage{tikz}
\usepackage{tikz-qtree,tikz-qtree-compat}
\usepackage{pgfplots}

\usepackage{lipsum}
\usetikzlibrary{calc}
\usetikzlibrary{patterns}

\usepackage{enumitem}
\graphicspath{{figures/}}

\usetikzlibrary{arrows,shapes}
\newcommand{\E}[1]{{\mathbb E}\left[ #1 \right]}

\newtheorem{Theorem}{Theorem}

\hypersetup{
	colorlinks,%
	citecolor=black,%
	filecolor=black,%
	linkcolor=black,%
	urlcolor=black
}

\pgfplotsset{
	grid style = {
		dash pattern = on 0.025mm off 0.95mm on 0.025mm off 0mm, 
		line cap = round,
		black,
		line width = 0.5pt
	},
	tick label style={font=\small},
	label style={font=\small},
	legend style={font=\footnotesize},
}

\pgfplotscreateplotcyclelist{laneas1}{
	cyan!60!black,		solid, 	every mark/.append style={fill=cyan!60!black},mark=x\\%
	cyan!60!black,		solid, 	every mark/.append style={fill=cyan!60!black},mark=+\\%
	cyan!60!black,		solid, 	every mark/.append style={fill=cyan!60!black},mark=o\\%
	red!80!black,       dashed\\%
	cyan!60!black, 		densely dotted,	every mark/.append style={fill=cyan!80!black},mark=diamond*\\%
}

\newtheorem{proposition}{Proposition}

\newtheorem{assumption}{Assumption}
\newtheorem{remark}{Remark}

\begin{document}

	\title{{  Heterogeneous Doppler Spread-based CSI Estimation Planning for TDD  Massive MIMO }}
	
	\author{\IEEEauthorblockN{$\text{Salah Eddine Hajri}^*, \text{Maialen Larra\~naga}^*, 
			̃ \text{Mohamad Assaad}^* $\\}
		\IEEEauthorblockA{*TCL Chair on 5G, Laboratoire des Signaux et Systemes (L2S, CNRS), CentraleSupelec, 91190	Gif-sur-Yvette, France\\
	$\{$Salaheddine.hajri,\;Maialen.Larranaga,\; Mohamad.Assaad$\}@$centralesupelec.fr}

		\thanks{Part of this work has been submitted to the  25th  international conference on Telecommunication \cite{Conf_version}.}
		
	}
	\maketitle	
	\begin{abstract}

Massive  multi-input  multi-output  (Massive MIMO) has been recognized as a key technology to meet the demand for higher data capacity and massive connectivity. Nevertheless, the number of active users is restricted due to training overhead and the limited coherence time. Current wireless systems assume the same coherence slot duration for all users, regardless of their heterogeneous Doppler spreads. In this paper, we exploit this neglected degree of freedom in addressing the training overhead bottleneck. We propose a new uplink training scheme where the periodicity of pilot transmission differs among users based on their actual channel coherence times. Since the changes in the wireless channel are, primarily, due to movement, uplink training decisions are optimized, over long time periods, while considering the evolution of the users channels and locations. Owing to the different rates of the wireless channel and location evolution, a two time scale control problem is formulated. In the fast time scale, an optimal training policy is derived by choosing which users are requested to send their pilots. In the slow time scale, location estimation decisions are optimized. Simulation results show that the derived training policies provide a considerable improvement of the cumulative average spectral efficiency even with partial location knowledge.

	\end{abstract}

\begin{IEEEkeywords}
	Massive MIMO,  Doppler spread, CSI estimation planning, Machine learning  
\end{IEEEkeywords}
	\IEEEpeerreviewmaketitle
	
	\pagebreak

	\section{Introduction}\label{sec:intro}

	Future wireless networks have to address an exponentially increasing demand for high data-rate.  In this context, several technologies have been proposed to improve the overall wireless networks performance: dense small-cell deployment, millimeter-wave communications and massive MIMO among others \cite{gupta2015survey}. 
	Massive MIMO  was identified as one of the most promising technologies to meet this requirement. Originally introduced by Marzetta \cite{Noncooperative}, massive MIMO  exploits a large number of  base station (BS) antennas in order to  enable  the  spatial  multiplexing  of  a  large  number  of devices. 
	By coherent processing of the signals over the BS antennas, transmit precoding  can  be  used  in order to concentrate  each  signal at  its  intended  terminal  and  receive  combining  can  be  used in order to discriminate between the signals of different users. Massive MIMO have been thoroughly studied and have  shown to improve the networks spectral efficiency (SE) and energy efficiency (EE) in addition to providing a  high capacity per area \cite{EEMIMO}. These gains are conditioned by an accurate channel state information (CSI) at the BSs. In this paper we will focus on Time Division Duplexing (TDD) systems, where CSI can be acquired using uplink training with orthogonal pilot sequences \cite{pilot_reduction1}. A major issue in TDD systems is that a number of these pilot sequences are reused resulting in {\it pilot contamination} \cite{Noncooperative, pilot_reduction1, pilot_reduction2}.  Another reason for  CSI inaccuracy is {\it channel aging}. This phenomenon results from  the variation of the channel between the instant when it is learned  and the instant when it is used for signal processing. This   time variation is  due  to  users mobility and processing delays at the BS.
	
	Performance degradation due  to channel aging was studied in a MIMO system with  coordinated  multi-point transmission/reception (CoMP) in \cite{aging1}. The authors showed  that  the impact of channel  aging is  mitigated  when	utilizing channel prediction filters  in the low  mobility  regime. The authors in Truong et al. \cite{aging2} provide an analysis of rate performance  in the presence  of channel aging and prediction.  They showed that, although channel aging  leads to degradation in the performance of massive MIMO systems,  channel prediction can overcome  this issue.	In  Papazafeiropoulos et al. \cite{aging3,aging4}, the effect of channel aging combined with channel prediction has been investigated in scenarios with  regularized Zero Forcing (ZF) precoders and minimum-mean-square-error (MMSE) receivers, respectively.  In Kong et al. \cite{aging5},  lower bounds of the sum-rate for both Maximum Ratio Combining (MRC) and ZF receivers with/without channel prediction have been derived with an  arbitrary number of BS antennas  and  users. The impact of channel aging  and  prediction on the power scaling law has been studied. The authors demonstrated that the transmit power scaling is not affected neither by  aged CSI nor  channel prediction.

	Channel aging  can also  be leveraged in order to optimize uplink  training. In  Vu et al. \cite{icc},   two spectral-efficient multiuser models for massive MIMO systems have been proposed. The main idea comes from the  observation that  users with low velocity are not required to send  training sequences with the same periodicity as faster moving users owing to the resulting heterogeneous coherence times. The  two proposed models proved to  achieve significant SE gains. 
	
	In this paper, we  aim  at  increasing the  SE  by exploiting  the heterogeneous  channel  aging  among users.  In the current  literature, the number of   scheduled users,  is limited by the fixed length  of the uplink training reference signal. A more appropriate approach would be to  define the  needed training resources dynamically, at each time slot. We aim at  adapting uplink training based on the actual  coherence times. This means that, at a given slot, if the correlation between the  estimated CSI  and  the  actual  channel  was not  considerably degraded, due to aging, the network  is not  required to reestimate it. Doing so enables to spear part of the  training resources that can  be used for data transmission or to schedule more users.  This is in accordance with  the concept of dynamic TDD that is already considered in the development  of the 5G standard \cite{qualcomm}.
	
	Channel  aging  results, primarily, from mobility, with speed being an important  parameter. Consequently, developing an uplink  training policy that takes into  consideration the  second order channel statistics  is of paramount importance. Developing such policy  requires accurate estimates of user locations, which can be rather complicated to  obtain, in practice. In fact, localizing all scheduled users requires non negligible signaling, if it is done through the localization capabilities of the network (OTDOA \cite{4G} for example).   Global Positioning System (GPS) can also be used  but it rises the problem  of the  life span of mobile devices batteries \cite{GPS_energy}. Consequently, we  suppose that  the network is able  to  estimate the location  of a limited set of users. Adapting to  the change in the large-scale fading coefficients and optimizing uplink  training decisions based on the  channel's autocorrelation should occur on  two different  time scales \cite{two_time_scale}. In fact the  two  optimizations are based on   information that change  over heterogeneous  time scales.  In order to achieve  the maximum  cumulative average SE  over time spans larger than  the large-scale fading coherence block, a two time  scale control problem  is considered. 
	
	In the fast time scale, an optimal training policy is derived. By taking into consideration the evolution   over time  of the correlation between the  estimated  CSI and the actual  channel,  the network is able  to  optimize its decisions to  schedule  users  for  uplink training   over  a finite time horizon. Taking into  consideration the time dimension allows the network to be more efficient since it becomes able to predict the impact of its decisions on long term performance. Deriving such policy can naturally be formulated as a discrete planning problem over a finite time horizon \cite{finite_hor}. The optimal  training  decisions  are derived for  a predefined time  duration, denoted here by $H$,  for  which the  large-scale  fading coefficients are  supposed to  be  constant. This is quite advantageous since it allows to  optimize training over time without requiring the  actual  channel estimates. Results prove that the  derived training policy  provides substantial performance increase. Since deriving the optimal  policy can be  computationally  prohibitive for  large optimization  horizons,  we provide a combinatorial optimization framework that  enables to  derive an approximate training policy with reduced  running time.

	In  the slow time scale, the network adapts to user mobility by deciding which users are required to  feedback their locations.  Estimating the exact location of all users requires a non negligible signaling overhead. Consequently,  efficiently selecting the  users that are required to feedback their location is important. Since locations are estimated in a periodic manner, we consider user locations that evolve  according to independent Markovian stochastic processes \cite{markov_model}. The location estimation problem introduced above, with locations evolving in a Markovian fashion, can be formulated as  a Partially Observable Markov Decision Process (POMDP)\cite{POMDP}.  Simulations prove that the combined optimization, on the two time scales,  provides an efficient training  strategy that improves the  achievable cumulative average SE even  with  partially erroneous geolocalisation.
	
	
	This paper is organized as follows. We describe the considered system model in Section II.  We discuss the advantages of coherence time  based training  in Section III. Two time scale training strategy learning is discussed in Section IV. Finally, in Section  V, numerical results are presented.

\section{System Model And Preliminaries}	\label{sec:model}

We consider  the uplink of a multi-cell multiuser massive MIMO system constituted of $C$ macro BSs operating in TDD mode.
Each macro BS is equipped with  $M$ omnidirectional antennas  and  serves $K$ mobile devices equipped, each,  with a	single omnidirectional antenna.  We will refer to the latter as users. All users in the network move according to different speeds and  directions. Consequently, their  signals are  subject to    heterogeneous Doppler spreads which  results in  different wireless channel autocorrelations  in time. We consider a system where  time  is slotted $t\in\{0,1,\ldots\}$ and the duration of each time slot $t$ is given by $D_c$. We note that $D_c$ is the channel  coherence time which depends on the maximum Doppler spread supported by the network, see for instance Toufik et al. \cite{4G}. We also  consider  the corresponding coherence  interval $T_s$. The wireless channel  of each user can be  decomposed as a product of   small and  large scale fading coefficients. The wireless channel from user $k$ (in cell $c$) to BS $j$,  at time slot $t$, i.e., $g^{[j]}_{ kc }(t)$, is given by
\begin{align}\label{eq:channel1}
& g^{[j]}_{ kc }(t) = 
\sqrt{\beta^{[j]}_{ kc}} h^{[j]}_{kc}(t), \text{ for all } k=1,\ldots, K, \text{ and } 
j,c=1,\ldots,C, 
\end{align}
where $h^{[j]}_{ kc }(t) \in \mathbb{C}^{ M \times 1}$ is the  fast
fading vector, $h^{[j]}_{ kc }(t) \sim \mathcal{C}\mathcal{N}(0,I_M)$. ${\beta^{[j]}_{ kc}} \in \mathbb R^+$ models the large-scale effect including shadowing and pathloss, which are assumed to remain constant  during large-scale coherence blocks of $T_{\beta}$ OFDM  symbols. 
\begin{remark}
	In Sections~IV and~V we will consider a system where $\beta^{[j]}_{kc}$ evolves according to a Markovian model.
\end{remark}	
\subsection{Channel  Estimation}

As introduced above, in this paper we focus on a TDD system, where the entire frequency band is used for downlink and uplink transmission by all BSs and users. 
The BSs acquire CSI  estimates  using orthonormal training sequences (i.e., pilot sequences) in the uplink. We consider a pilot reuse factor of $1$, i.e.  the same sets of pilot sequences are used in all cells.

We also consider that, during each  coherence  interval, a maximum of  $\tau$ users are scheduled for uplink training in each cell with $\tau \leq K$. For that,  we consider a set of orthonormal training sequences, that is, sequences $q_i\in \mathbb C^{\tau \times 1}$ such that $ q^{\dagger}_i q_j =\delta_{ij} $ (with $\delta_{ij}$ the Kronecker delta). 

During  uplink training  of slot $t$, the $l^{th}$ BS receives the pilot signal $	Y^{[l]}_p(t) \in \mathbb{C}^{ M \times \tau}$
\begin{align}\label{eq:training_signal}
Y^{[l]}_p(t) &=  \sum_{c=1}^{C} \sum_{k=1}^{\tau} \sqrt{P_p} g^{[l]}_{ kc } (t) q^\dagger_k + W_p(t), 
\end{align}
where $W_p(t)  \in \mathbb{C}^{ M \times \tau} $ refers to  an additive  white  Gaussian  noise matrix with i.i.d. $\mathcal{C}\mathcal{N}(0,1)$ entries. $P_p$ refers to  the training  signal power.
The  $l^{th}$ BS  then uses the  orthogonality of  training sequences in order to  obtain  the MMSE  estimate of the  channel of user $k,l$ \cite{pilot_reduction2} as
\begin{align}\label{eq:channel_estimation}
\hat{g}_{kl}^{[l]}(t) =   \frac{\beta^{[l]}_{ kl}}{\frac{1}{P_p}+\sum_{b, b\neq l}^{C} \beta^{[l]}_{ kb}}    \frac{Y^{[l]}_p(t)}{\sqrt{P_p}} q_k .
\end{align}
Note that the MMSE channel estimate $\hat{g}^{[l]}_{ kl } (t) $  follows a  $\mathcal{C}\mathcal{N} \left(0,\frac{\beta^{[l]^2}_{ kl}}{\frac{1}{P_p}+\sum_{b, b\neq l}^{C} \beta^{[l]}_{ kb}} I_M \right) $ distribution. The  wireless channel between user $k$ (in cell $l$) and BS $l$ can then be decomposed as follows
\begin{align}\label{eq:channel_error_estimation}
g_{kl}^{[l]}(t) = \hat g_{kl}^{[l]}(t)+\tilde g_{kl}^{[l]}(t),
\end{align}
where $\tilde g_{kl}^{[l]}(t)$ represents the  estimation error  and  follows a $ \mathcal{C}\mathcal{N} \left(0,\left(\beta^{[l]}_{ kl}- \frac{\beta^{[l]^2}_{ kl}}{\frac{1}{P_p}+\sum_{b, b\neq l}^{C} \beta^{[l]}_{ kb}}\right) I_M \right) $ distribution.
Moreover, $\hat{g}^{[l]}_{ kl } (t) $ and $\tilde{g}^{[l]}_{ kl } (t)$ are independent \cite{pilot_reduction2}.

\subsection{Channel  aging}		

In practice, the wireless channel varies between the time when it is learned and used for precoding in downlink and decoding in uplink. This variation  is due mainly to  user movement and processing delays. Such phenomenon is referred to as {\it channel aging}. 
Its impact can be captured  by a time varying wireless channel model. To this end we consider a stationary ergodic Gauss-Markov block fading regular process (or auto-regressive model of order~1) \cite{jakes}. The evolution of the channel vector of user $k,l$ between the two slots $t$ and $t-1$ is expressed as
\begin{align}\label{eq:channel_evolution}
&  g^{[l]}_{ kl } (t) = 
\rho_{ kl }^{[l]} {g}^{[l]}_{ kl} (t-1) 	 + \sqrt{\beta^{[l]}_{kl}} \varepsilon^{[l]}_{ kl }(t), 
\end{align}
where  $ \varepsilon^{[l]}_{ kl }(t) $  denotes a  temporally uncorrelated complex white Gaussian noise process with zero mean and  variance  $(1- \rho^{[l]^2}_{ kl }) I_M$.
$\rho^{[l]}_{ kl }$  represents a temporal correlation parameter of the channel of user $k,l$. This parameter is given by Jakes et al. \cite{jakes} and reads as follows
\begin{align}\label{eq:rho}
& \rho^{[j]}_{ kl } = J_0(2\pi f^{[j]}_{ kl } D_c),
\end{align}
where $J_0(\cdot)$ is the zeroth-order Bessel function of the first kind and  $f^{[j]}_{ kl }$  represents the maximum Doppler shift of user $k$ in cell $l$ with respect to the antennas of BS $j$. In our work,  we adopt a realistic setting in which, mobile users have different frequency shifts  since we consider heterogeneous movement velocities and directions. For every user $k$ in cell $l$, the maximum Doppler shift with respect to the antennas of BS $j$ is given by
\begin{align}\label{eq:Doppler}
& f^{[j]}_{ kl } = \frac{\nu_{kl}f_{\textit{c}}}{\textit{c}} \text{cos}(\theta^{[j]}_{kl}), 
\end{align}
where $\nu_{kl}$ is the velocity of user $k$ in cell $l$ in meters per seconds, $\textit{c}= 3 \times 10^8 \; \text{mps} $ is the speed of light,  $f_{\textit{c}}$ is the carrier frequency and $\theta^{[j]}_{kl}$ represents the angular difference between the directions of   the mobile device  movement and the incident wave. 
Taking into consideration the combined effects of estimation error and impairments due to channel aging, we can express the wireless channel  of user $k,l$ at time $t$ as 
\begin{align}\label{eq:channel_decomp}
&  g^{[l]}_{ kl } (t) = 
\rho_{ kl }^{[l]} \hat{g}^{[l]}_{ kl } (t-1) +	\rho_{ kl }^{[l]} \tilde{g}^{[l]}_{ kl } (t-1) 	 + \sqrt{\beta^{[l]}_{kl}} \varepsilon^{[l]}_{ kl}(t), 
\end{align}

	\section{An adaptive uplink training approach for Massive MIMO TDD systems}
	In current Massive MIMO models, the same  coherence  interval $T_s$ is considered for  all users.  $T_s$ is defined as a system parameter that  is based on the  maximum  Doppler spread supported by the  network \cite{4G}. This  consideration results in a suboptimal  use of the  time-frequency resources and a loss of flexibility that can be leveraged otherwise. In fact, in practice, users experience heterogeneous Doppler spreads. Consequently, their channels do not age at the  same  rate. Forcing all users  to  perform uplink training  with the same  periodicity causes  vain redundancy and a loss of resources. A more  efficient approach should adapt the  periodicity of each user CSI  estimation  according to its actual  coherence  time \cite{Conf_version} \cite{icc}. This means that, at a given slot, if the correlation between the  estimated CSI  and  the  actual  channel  was not  considerably degraded,  the network  is not  required to reestimate it. Doing so enables to spear part of the  training resources that can  be used for data transmission or to schedule more users. In all cases, the latter results in  an increase in   SE. In this section, we present a novel approach for uplink training   that leverages the users heterogeneous  channel coherence times. We present  a detailed analysis of its impact  on the  achievable SE  with  MRC receivers.  We also derive an important  condition which  ensures that the proposed scheme is able to improve performance. 
	
	\subsection{An adaptive coherence time-based uplink training scheme}
	We consider a massive MIMO system in which  CSI estimation is adapted according  to  the  actual  users'  coherence times. We consider that  the network groups users according to their channel autocorrelation coefficients into $N_G$ copilot user groups $\lambda_g, \; g=1,...,N_G$. The users in each group  are either scheduled for uplink training synchronously, using the same pilot sequence, or not scheduled at all. This requirement guarantees that copilot users always have the same CSI delay and a similar channel  aging effect. For each  copilot  group $\lambda_g, \; g=1,...,N_G$,  the CSI delays are denoted by  $d_g, g=1,...,N_G$. At each  slot, all $N_G$ copilot user groups  are scheduled for data transmission and a maximum of  $\tau \;(\tau<N_G)$ copilot groups are selected for uplink training. The  rest will  have their signals processed using  the last estimated version of their CSI.
	The proposed  Time-Division  Duplexing  (TDD)  protocol  consists  of the following seven steps.

	\begin{enumerate}
		
		\item In the beginning of each  large-scale coherence block, the BSs estimate the large scale fading  and  channel  autocorrelation coefficients, i.e., $\beta_{kc}^{[j]}$ and $\rho_{kc}^{[j]}$ for all $k=1,\ldots,K$, and $c,j=1,\ldots,C$. All coefficients are then fed back to a central  processing unit (CPU).
		
		\item Next, the CP clusters users according to their autocorrelation coefficients using  $K$-mean, see Young et al~\cite{mean}. 
		The resulting clusters will  be characterized by an average autocorrelation coefficient or, equivalently, an average Doppler spread and a variance of the  corresponding users autocorrelation coefficients.
		The considered number of clusters is $N_c$. Defining  $N_c$  is of paramount importance. In this work, we choose to define $N_c$ according to
		\begin{align}\label{eq:cluster_num}
		N_c= \lceil\frac{D_{max}}{D_c}\rceil,
		\end{align}
		where $D_{max}$ represents the maximum coherence time. $(9)$ guarantees that the average coherence time per cluster is approximately equivalent to a multiple of  $D_c$. This is needed in order to  appropriately define CSI  estimation periodicity  as a function of  the  parameter $D_c$.

		\item Next, the  CP allocates all users in  the network ($K$ per cell) to $N_G$ copilot groups. Each  group contains at maximum $C$ users from the same channel autocorrelation based cluster and from different cells.  These $N_G$ copilot groups are formed with  minimum  variance of the autocorrelation coefficients in each  group. This  guarantees that  copilot users has similar channel aging impact. The justification for this grouping is discussed in  Section III.C. 

		\item At each coherence  slot, the network  schedules at maximum $\tau$ copilot  groups  for uplink training synchronously. Depending on the main key performance indicator (KPI) to  optimize,   different scheduling algorithms can be  used to select these copilot  groups.	In this  paper, we  propose  a scheduling algorithm  that  exploits the  aforementioned user grouping in order to derive an optimal CSI estimation policy. This is  the focus of  Section IV  and represents one of the main contributions of the  present work.
		\item All  $N_G$  copilot  groups  transmit their  uplink signal in  a synchronous manner.
		\item  The BSs process the received pilot signal and estimates the channels of the  active users during uplink training using MMSE estimators. The BSs decode and precode the uplink and downlink  data signals, respectively, using the last estimated version of each  user CSI.
		\item All BSs synchronously transmit  downlink data signals to the $N_G$  copilot  groups. 
	\end{enumerate}

	\subsection{Spectral efficiency with  outdated CSI}

	In what follows,  we analyze the impact of  the aforementioned training procedure on the achievable SE with a MRC receiver. We also explain why an adaptive Doppler based training can be more efficient.  Moreover, we provide a condition in order to ensure that the   spectral  efficiency of  all users is improved when the aforementioned training procedure is used. For the sake of analytical traceability, we consider that the $N_G$ copilot groups contain exactly $C$ users.	We, henceforth, refer to each user  by its copilot group and serving BS indexes. During uplink data transmission, at time slot $t$,  BS $l$ receives the data signal $Y^{[l]}_u(t)$ which is given by					
	\begin{align}\label{eq:uplink_signal_data}
	Y^{[l]}_u(t) =  \sum_{c=1}^{C} \sum_{k=1}^{N_G} \sqrt{P_u} g^{[l]}_{ kc } (t) S_{kc} +W_u(t), 
	\end{align}
	where  $W_u(t) \sim CN(0,I_M)$ is the additive noise, $S_{kc}$ denotes the  uplink signal of  user $k,c,\;k=1,\ldots,N_G,\; c=1,\ldots,C$  and $P_u$ denotes  the reverse link  transmit power. Each BS applies a MRC receiver based on the  latest available  CSI  estimates.  BS  $l, l=1,\ldots,C$ detects the signal of user $g, g=1,\ldots,N_G$, within the same cell, by applying the following
	
	\begin{align}\label{eq:mf}
	& u_{gl}(t)= \frac{\hat{g}^{[l]}_{ gl } (t-d_{g}) }{\lVert \hat{g}^{[l]}_{ gl } (t-d_{g})  \rVert}, t\geq d_{g}, 
	\end{align}
	where $\hat{g}^{[l]}_{ gl } (t-d_{g})$ denotes the latest available channel estimate  of user $g$ in cell $l$.
	The resulting average achievable SE  in the system with MRC receivers is given in  Theorem $1$.

	\begin{Theorem}
		
		For $N_G$ active copilot groups,  $\tau$  of which are scheduled for uplink training and using  a MRC receiver $u_{gl}(t)$ that is based on the latest available CSI estimates of each user $g,l$, the average achievable spectral efficiency in the uplink $\bar{R}^{MRC}_u $ is lower bounded by:
		
		\begin{align}\label{eq:average_spectral}
		& \bar{R}^{MRC}_u \geq \sum_{l=1}^{C} \sum_{g=1}^{N_G}
		\left( 1-\frac{\tau}{T_s}\right) \text{log} \left(1+ \frac{(M-1)\beta^{[l]^2}_{gl} \rho^{[l]^{2 d_g}}_{ gl } }{(M-1) \times I^p_{gl} +  I^n_{gl}  } \right), 
		\end{align}
		where $d_g, g=1...N_G$ represents the copilot groups CSI  delays. $I^p_{gl}$ and $I^n_{gl}$ are given by:
		\begin{align}\label{eq:interference}
		& I^p_{gl} = \sum_{ c\neq l}^{C} \rho^{[l]^{2 d_g}}_{ gc } \beta^{[l]^2}_{gc}, \\
		& I^n_{gl}= (\sum_{c=1}^{C} \sum_{k\neq g}^{N_G} \beta^{[l]}_{kc} +  \sum_{ c=1}^{C} ( \beta^{[l]}_{gc} -    \rho^{[l]^{2 d_g}}_{ gc }\frac{ \beta^{[l]^2}_{gc}}{\frac{1}{P_{p}}+\sum_{b=1}^{C}  \beta^{[l]}_{gb} }) + \frac{1}{P_u})\times (\frac{1}{P_{p}}+ \sum_{b=1}^{C} \beta^{[l]}_{gb}  ).  \nonumber
		\end{align}	
	\end{Theorem}
	\begin{IEEEproof}
		See appendix A.
	\end{IEEEproof}

	Equation~\eqref{eq:average_spectral} provides further insights into the impact of  channel aging on the achievable average SE as a function of the CSI time offset. We can clearly see that the  SE   decreases as a function of its CSI time offset. This is an intuitive result since the correlation between the estimated CSI and the actual channel  fades over time. Equation~\eqref{eq:average_spectral} shows also that for a same CSI time offset, the degradation due to channel aging  is  higher for  users with  lower autocorrelation coefficients. 
	Although  outdated CSI  causes an SINR degradation, the speared resources from  uplink training can lead to  an increase in  SE. 
	\subsection{ASYMPTOTIC Performance}
	
	We now analyze  the potential gain that the proposed training approach can provide. To do so,  we compare it with a reference model that  follows a   classical TDD protocol in which all of $N_G$ copilot groups are scheduled for uplink training at each time slot. We consider a worst case scenario with random delays and random copilot groups allocation. In this scenario,  each user experiences the lowest channel autocorrelation coefficient  in comparison with its copilot users. This means that  each user suffers from the heaviest channel aging impact in its copilot group. 
	
	\begin{Theorem}
		In the asymptotic regime ($M$ grows large), with  $\bar{\rho}^{[{min}]}_{g}$ and $\bar{\rho}^{[{max}]}_{{g}}$ denoting, respectively, the minimum and maximum  autocorrelation  coefficients in copilot group $g,\; g=1,...,N_G$,  the proposed training framework enables to improve the SE of each user when~\eqref{eq:cond} is satisfied
		\begin{align}\label{eq:cond}
		\left( \frac{\bar{\rho}^{[{min}]^2}_{g}}{\bar{\rho}^{[{max}]^2}_{{g}} }\right)^{d_g} \geq \frac{ \left( 1+\text{ SINR}^{[\infty]}_{g,l}\right)^{\frac{T_s-N_G}{T_s-\tau}}-1}{\text{ SINR}^{[\infty]}_{g,l}}, 
		\end{align}
		with
		\begin{align}\label{eq:cond_sinr}
		\text{ SINR}^{[\infty]}_{g,l}= \frac{\beta^{[l]^2}_{gl}  }{	\sum_{ c\neq l}^{}  \beta^{[l]^2}_{gc}},  
		\end{align}	
	\end{Theorem}
	\begin{IEEEproof}
		See appendix B.
	\end{IEEEproof}
	Condition~\eqref{eq:cond} ensures that the SE of  each user increases when outdated CSI is  used. Equation~\eqref{eq:cond} shows that the speared resources due  to the reduced training overhead is  a defining  parameter. 
	In fact,  SE is improved as long as the SINR degradation is compensated for by  the spared resources from training. It also shows the importance of  the  ratio between the  minimum and maximum   autocorrelation coefficients in a copilot group. A high  ratio is required in order to  achieve  the needed SE gain. This  requirement become  tighter as the  CSI  time  offset increases. \eqref{eq:cond}  shows that the use of the proposed procedure  can improves the achievable SE even with random delays and  random pilot sequence allocation.

	\begin{remark} 
		In order to  satisfy condition~\eqref{eq:cond},  copilot users  need to have similar  autocorrelation coefficients. This explains  Steps $2)$ and $3)$ in the protocol in Section~III.A. In fact,  clustering  users based on   their autocorrelation  coefficients and grouping them  accordingly  results in   copilot  user groups  with  homogeneous channel  aging  within  each  group. This allows to  tolerate higher CSI  time offset.  \eqref{eq:cond} also shows that the use of the aforementioned training procedure  can improves the achievable SE of the  network, even with random pilot allocation. Consequently, one can do better if a coherence time adaptive scheduling for uplink training  is implemented.	
		More importantly, the proposed scheme shows the impact of the time dimension. This fact justifies the need for a time-aware  training optimization which will be  the focus of the next section.   
	\end{remark}

	\section{ Optimal training strategy with outdated CSI  and  user mobility: a  two-time scale decision process}
	
	We proved that adapting uplink training  periodicity  to the  actual channel coherence time can provide a considerable increase in network performance, even with random pilot sequence allocation.  Nevertheless, higher performance gain can be obtained if more sophisticated and adapted scheduling policy is used. Developing such policy is the focus of this section.
	
	As a matter of fact, knowing that  CSI  estimation periodicity  should depend on the  rate of channel aging, it makes sense to develop  an  uplink training  policy  that  takes into  consideration  the evolution in the difference  between the  estimated CSI and the actual  wireless channels.
	In opposition to a per slot  uplink  training optimizing,   such  policy  enables to  take  into  consideration  the  impact of  past scheduling decisions  on the  long term performance. User mobility should also be included. In fact,  channel  aging  results, primarily, from mobility, with velocity being a defining parameter. Consequently, developing an uplink  training policy that takes into  consideration the  evolution of  large-scale fading coefficients, in addition to channel  aging,  is of paramount importance. Developing such strategy  requires accurate estimates of user locations, which can be rather complicated to  obtain, in practice. As a matter of fact, localizing all covered users requires a non negligible signaling overhead and energy consumption \cite{4G}, \cite{GPS_energy}. Consequently, this problem should be addressed while assuming  a partial  knowledge of the user positions.    
	Adapting to  the change in user locations and optimizing uplink  training decisions based on the  channels' autocorrelation coefficients, should occur on  two different  time scales \cite{two_time_scale}. In fact, the  two  optimizations are based on   information that change  over heterogeneous  time scales (The wireless channel  changes faster than user position). Consequently, a two time  scale control problem should be formulated. This will be the focus of the present section.
	
	\subsection{ Optimizing uplink training: A two-time scale control problem}
	
	We now model  the  two-time scale system introduced above as a POMDP \cite{POMDP}. We  assume finite action and state spaces in both time scales (see Figure~\ref{fig:two-time-scales-model}). 
	\begin{figure}
		\centering
		\includegraphics[width=12cm,height=4cm]{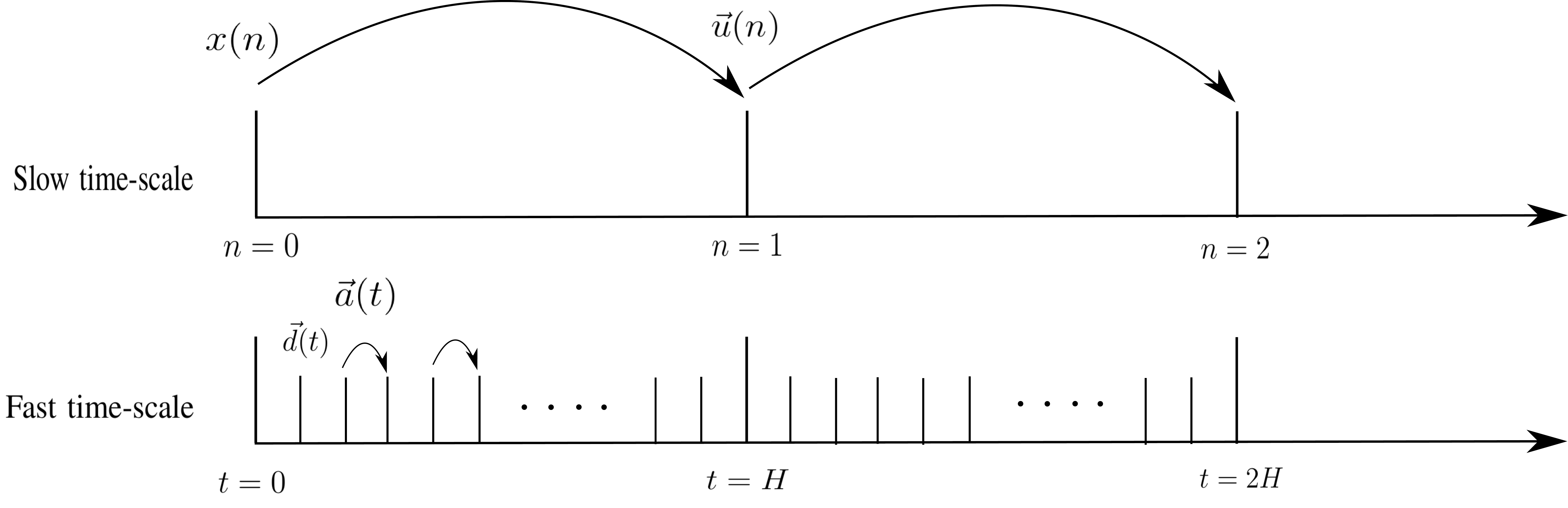}
		\caption{A two time-scale planning problem}\label{fig:two-time-scales-model}
	\end{figure}
	
	We consider that, in the slow-time scale ({\it upper level}), the position of the users evolves according to a Markovian Mobility model \cite{markov_model} within its serving cell. These position variations occur at decision times $n=0,1,\ldots$.	Let $\ell_{g}(n)$ be the combination of the positions of users from copilot group $g$  at time $n$. 
	Considering the combination of  copilot users positions  instead of  each individual one enables to reduce the complexity  of  the present  model. We  assume, for the sake of simplicity,  that all copilot groups have $L$ possible position combinations, hence $\ell_{g}(n)\in\{1,\ldots,L\}$.
	Building this model requires a portioning of the coverage area of each cell into a number of disjoint regions. The area of each  region is chosen such  that  the variation  of the large scale  fading coefficients can be  considered as negligible within the  region.
	For copilot group  $g$, each  position $\ell_{g}(n)\in\{1,\ldots,L\}$ corresponds to a combination of  regions in each cell.
	The transition probabilities are characterized by the matrix
	\begin{align}\label{eq:transition_matrix}
	P_{g}=(p_{g}(i,j))_{i,j\in\{1,\ldots,L\}}, \hbox{  for copilot group } g. 
	\end{align}
	The large scale fading coefficients for user $g$ in cell $l$, i.e.,  $\beta_{gl}^{[j]}, j\in\{1,\ldots,C\}$ depend on the users' position. In previous sections, we assumed that this values were constant. In this section, we add a time dependency to it, namely, $\beta_{gl}^{[j]}(n)=\beta_{gl}^{[j],\ell_{g}(n)}\in\{\beta_{gl}^{[j],1},\ldots,\beta_{gl}^{[j],L}\}$. Acquiring the information on the position of all users can be really expensive in terms of  processing overhead and  energy consumption \cite{4G},\cite{GPS_energy}.
	Consequently, we  consider that a limited number of users can feedback its positions to the network. In particular, we assume that,  in every decision epoch, the users from $U_{max}$ copilot groups can feedback their positions (with $U_{max}<N_G$). The CP therefore can only acquire the positions of the users from $U_{max}$ copilot groups, at each time $n$. The positions of the rest of the  user will be inferred from previous estimations. This estimation is characterized by the belief state vector. The belief state vector of copilot group $g$, at decision-time $n$, will be denoted by $\vec b_{g}(n)$, where the $i^{th}$ entry in  $\vec b_{g}(n)$ refers to the probability that the users of copilot group $g$ are in positions of combination $i$. We define by $\mathcal{X}_{g}$ the set of all belief states for copilot group $g$  and we let $\mathcal{X}=\mathcal{X}_{1}\times\ldots\times\mathcal{X}_{N_G}$ be the state space in the upper level. A remark on the notation is now in order.
	
	\begin{remark} The state in the upper level $x\in\mathcal{X}$ is an $L\times N_G$ matrix, whose columns represent the belief state vectors of all copilot groups $g$, for $g=1,\ldots,N_G$. That is, $x = (\vec b_{1}, \ldots,\vec b_{N_G})$.
	\end{remark}
	
	In the upper level,  at every decision epoch $n=0,1,\ldots$, the decision is to select which $U_{max}$ copilot groups out of the $N_G$  will transmit their positions to the BSs. That is, we consider the action vector $\vec u (n)=(u_{1}(n),\ldots,u_{N_G}(n))\in\mathcal{A}=\{0,1\}^{N_G}$, such that
	\begin{equation}\label{eq:action_space}
	\begin{aligned}
	& u_{g}(n)=
	\begin{cases}
	1 & \mbox{ users in copilot group $g$  feedback their positions at decision epoch $n$,}  \\
	0 & \mbox{otherwise. } 
	\end{cases}
	\end{aligned}
	\end{equation}
	
	At decision epoch $n$,  the transition probability from belief state matrix $x(n)\in\mathcal{X}$ to belief state matrix $x(n+1)\in\mathcal{X}$ is defined by
	\begin{align}\label{eq:belief_transition}
	\mathbb{P}^{up}(x(n+1)=x'|x(n)=x,\vec u(n))=& \mathbb{P}(\vec b_{1}(n+1)=b_{1}'|\vec b_{1}(n)=b_{1},\vec u(n))\cdot\ldots\\
	&\cdot\mathbb{P}(\vec b_{N_G}(n+1)=b_{N_G}'|\vec b_{N_G}(n)=b_{N_G},\vec u(n)), \nonumber
	\end{align} 
	where, $x'=(\vec b_{1}',\ldots,\vec b_{N_G}')$, $x=(\vec b_{1},\ldots,\vec b_{N_G})$ with $b_{g}',b_{g}\in\mathcal{X}_{g}$ for all $g=1,\ldots,N_G$ and $\vec u(n)\in\mathcal{A}$. The latter is satisfied because all users have independent movements. Recall that each  position combination of users in copilot group  $g$ is characterized by a set of large scale fading coefficients $\beta_{gl}^{[j]}, j\in\{1,\ldots,C\},\; l\in\{1,\ldots,C\}$. In the fast-time scale, we define the state-space by $X=\{0,\ldots,H-1\}^{N_G}$, that is, the set of all possible delay vectors. Namely, $\vec d=(d_1,\ldots,d_{N_G})\in X$ is such that $d_g$ is the CSI delay of all users in copilot group $g$, i.e., $\lambda_g$. The action space is $A=\{0,1\}^{N_G}$. For $\vec a=(a_1,\ldots,a_{N_G})\in A$, $a_g, 1,\ldots,{N_G}$ is given by
	\begin{equation}\label{eq:action_space}
	\begin{aligned}
	& a_g=
	\begin{cases}
	1 & \mbox{ copilot group $g$ is scheduled for uplink training,}  \\
	0 & \mbox{otherwise. } 
	\end{cases}
	\end{aligned}
	\end{equation}
	The decision times at the fast-time scale ({\it lower level}) will be denoted by $t=\{t_0,t_1,\ldots\}$, with $t_{nH}=n$ for all $n=0,1,\ldots$ and $H$ the finite-time horizon in the {\it lower level}. 
	Moreover, we make the assumption that the decision $\vec u (n+1)$, in the slow-time scale, is made right after the decision at time $t_{nH}$. We denote by $\vec d(0)=\vec d^0\in X$ the initial state in the fast-time scale at $n=0$ and $x_0\in \mathcal{X}$ the initial state in the slow-time scale. In this particular model, the fast time scale transitions from time $t_{nH}$ until time $t_{(n+1)H-1}$ for all $n\geq 0$ are deterministic. Namely,
	\begin{align}\label{eq:deterministic_trans}
	d_g(t_{nH+j}) = (1+d_g(t_{nH+j-1}))(1-a_g(t_{nH+j})), \hbox{ for all } n\geq0, \hbox{ and } 1\leq j\leq H. 
	\end{align}
	At the fast time scale, we therefore encounter a finite-state finite-horizon deterministic sequential-decision problem \cite{finite_hor}. The reward in this level, at time $t$ with  MRC receivers, is the following
	\begin{align}\label{Rlow}
	R^{low}(\vec d (t),\vec a(t), x, \vec u) = \sum_{g=1}^{N_G}\sum_{l=1}^C\left(1-\frac{1}{T_s}\sum_{i=1}^{N_G} a_i(t)\right)\log\left(1+\hbox{SINR}_{gl}^{MRC}(\vec d(t), x,\vec u)\right),
	\end{align}
	where $x\in\mathcal{X}$ and $\vec u \in \mathcal{A}$ are fixed and
	\begin{align}\label{Reward_MF}
	\hbox{SINR}_{gl}^{MRC}(\vec d(t), x,\vec u) = \frac{(M-1)(\beta_{gl}^{[l]})^2(\rho_{gl}^{[l]})^{2d_{g}(t)}}{(M-1)\times I_{gl}^p + I_{gl}^n},
	\end{align}
	the SINR of user $g$ in cell $l$ with MRC receiver.  $I^p_{gl}$ and $I_{gl}^n$ are given in  Theorem $1$.\\ 
	Note that the reward function at the lower level, i.e., $R^{low}$, depends on the belief state and the decision in the upper level.
	We now define the sequence $\pi^{low}=\{\vec \phi^{low}_n\}_{n=0}^\infty$, where for each $n$, 
	\begin{align}\label{eq:phi_low}
	\vec\phi_n^{low}=(\phi_{t_{nH}}^{low},\phi^{low}_{t_{nH+1}},\ldots,\phi^{low}_{t_{(n+1)H-1}}).
	\end{align}
	Each function $\phi_{t_{nH+j}}^{low} : X\times\mathcal{X}\times\mathcal{A}\to A$ prescribes the action to be taken at decision time $t_{nH+j}$ (in the lower level), for all $n\geq0$ and all $0\leq j\leq H-1$. For this model we only look at the set of stationary decision rules, $\pi^{low}$ with respect to the upper level, such that  $\vec\phi_n^{low}(\vec d,x,\vec u)=\vec\phi_{n'}^{low}(\vec d,x,\vec u)$ for all $n$ and $n'$ given $\vec d\in X$, $x\in\mathcal{X}$ and $\vec u\in\mathcal{A}$. That is, for fixed $\vec d\in X$, $x\in\mathcal{X}$ and $\vec u\in\mathcal{A}$ the optimal decision rule in the lower level will be independent of the decision epoch $n$ in the slow time scale. This consideration is in accordance with most existing literature and can also be justified by the considered setting. 
	The set of all possible lower level decision rules will be denoted by $\Pi^{low}$, i.e., $\pi^{low}\in\Pi^{low}$.  Moreover, we drop the dependency on $n$, since we only consider policies that are $n$-independent, and we denote by $\Phi^{low}$ the set of all $H$-horizon policies $\vec \phi^{low}$, i.e., $\vec\phi^{low}\in\Phi^{low}$. We now define $\Phi^{low}_{x,\vec u}\subset \Phi^{low}$ as follows
	\begin{align}\label{eq:Phi_low}
	\Phi^{low}_{x,\vec u} = \{&\vec\phi^{low}_{x,\vec u}:\vec\phi^{low}_{x,\vec u}=(\phi^{low}_{x,\vec u,t_0},\ldots,\phi^{low}_{x,\vec u,t_{H-1}}), \phi^{low}_{x,\vec u,t_{j}}:X\times\{x\}\times\{\vec u\}\to A \hbox{ and } j = 0,\ldots,H-1\}.
	\end{align}
	The latter is the set of all $H$-horizon policies given initial belief state matrix $x$ and action in the upper level $\vec u$. Note that, in the definition of $\Phi^{low}_{x,\vec u}$ to introduce the policy $\vec \phi^{low}_{x,\vec u}$, we use the decision times $t_0,\ldots,t_{H-1}$. This is without loss of generality, since we recall that these policies are independent from $n$. Next we define the reward in the upper level. Namely,
	\begin{align}\label{problem:lower level}
	R^{up}(\vec d, \vec\phi^{low}, x(n),\vec u(n)) = \sum_{t=t_{nH}}^{t_{(n+1)H-1}}R^{low}(\vec d(t),\phi_{t}^{low}(\vec d(t),x(n),\vec u(n)), x(n), \vec u(n)),
	\end{align}
	where $\vec d$ is the delay state vector at time $t_{nH}$.
	We remark that none of the upper level decisions incur in an immediate cost. Let us denote by $\Phi^{up}$ the set of all possible stationary decision rules in the upper level, such that $\pi^{up}\in\Phi^{up}$, $\pi^{up}:X\times\mathcal{X}\to\mathcal{A}$. 
	Consequently,  the objective is to find $\pi^{up}\in\Phi^{up}$ and $\pi^{low}\in\Phi^{low}$ such that
	\begin{align}\label{problem:lower upper_level}
	\max_{\pi^{up}\in\Phi^{up}}\max_{\pi^{low}\in\Phi^{low}}\lim_{Z\to\infty}\frac{1}{Z}\sum_{n=0}^{Z-1}\mathbb{E}\left(R^{up}(\vec d(t_{nH}), \pi^{low}, x(n),\pi^{up}(\vec d(t_{nH}),x(n)))\right).
	\end{align}
	
	The latter problem is a POMDP \cite{POMDP}. To see this, it suffices to note that the slow time scale  sequential decision making problem is just a POMDP with a reward that depends on the fast time scale deterministic decision making problem. Therefore the standard theory on Bellman's optimality equations follows.  The optimal decision-rule for this POMDP can be obtained as a solution of the optimality equation for $0<\alpha<1$
	\begin{align}\label{eq:optimality_equation}
	& V(\vec d,x) = \max_{\vec u \in \mathcal{A}}\left(\max_{\vec\phi^{low}_{x,\vec u}\in\Phi^{low} }\{R^{up}(\vec d,\vec\phi^{low}_{x,\vec u},x,\vec u)+\alpha\sum_{y\in\mathcal{X}}\mathbb{P}^{up}(y|x,\vec u)V(\vec d^{\vec \phi^{low}_{x,\vec u}},y)\}\right).
	\end{align}
	$V(\vec d,x)$ denotes the value function \cite{POMDP} which refers, in our case, to the  long term   SE.
	We will now make an assumption that simplifies the model significantly. We define $\overline\Phi^{low}\subset\Phi^{low}$ where
	\begin{align}\label{eq:optimality_equation_simplification}
	\overline\Phi^{low}=\{&\vec \phi^{low}:\vec \phi^{low}=(\phi_{t_0}^{low},\ldots,\phi_{t_{H-1}^{low}}),\\\nonumber
	&\phi^{low}_{t_j}:X\times\{x\}\times\{\vec u\}\to A \hbox{ for } j = 0,\ldots, H-1, \hbox{ and } \phi^{low}_{t_{0}}=(1,\ldots,1)\}.
	\end{align}
	For all $\vec \phi^{low}\in\overline\Phi^{low}$, $\vec\phi^{low}$ is such that, in the first stage of the $H$-horizon problem,  all copilot groups are scheduled for uplink training. This allows us to start every slow-time scale with the same delay state $\vec d(nH)=(0,\ldots,0)$ for all $n=0,1,\ldots$.  Eq.~\eqref{eq:optimality_equation} then reduces to
	\begin{align}\label{eq:optimality_equation_tran}
	&V(x) = \max_{\vec u \in \mathcal{A}}\left(\max_{\vec\phi^{low}_{x,\vec u}\in\overline\Phi^{low} }\{R^{up}(\vec\phi^{low}_{x,\vec u},x,\vec u)\}+\alpha\sum_{y\in\mathcal{X}}\mathbb{P}^{up}(y|x,\vec u)V(y)\right),
	\end{align}
	where $R^{up}(\vec\phi^{low}_{x,\vec u},x,\vec u)=R^{up}((0,\ldots,0),\vec\phi^{low}_{x,\vec u},x,\vec u).$
	If we further denote 
	\begin{align}\label{eq:optimality_equation_max}
	R^{max}(x,\vec u)=\max_{\vec\phi^{low}_{x,\vec u}\in\overline\Phi^{low} }\{R^{up}(\vec\phi^{low}_{x,\vec u},x,\vec u)\},
	\end{align}
	we then obtain a standard one-time scale POMDP, and its optimality equation reduces to
	\begin{align}\label{eq:optimality_equation_standard}
	& V(x) = \max_{\vec u \in \mathcal{A}}\left(R^{max}(x,\vec u)+\alpha\sum_{y\in\mathcal{X}}\mathbb{P}^{up}(y|x,\vec u)V(y)\right).
	\end{align}
	
	POMDPs have been long studied in the literature. It was shown that the complexity of POMDP exact algorithms grows exponentially with the number of state variables \cite{POMDPreview}. Even for simpler finite-horizon POMDPs, finding the optimal policy is PSPACE-hard \cite{POMDPreview}.
	This means that deriving an optimal policy for \eqref{problem:lower upper_level} is too complex since  belief-state monitoring is infeasible for large problems. As \eqref{problem:lower upper_level}  is too complex to  solve directly, we  decompose the problem  and tackle the two time-scales  separately (see Figure \ref{fig:two-time-scales-model}).  Indeed in order to  solve \eqref{problem:lower upper_level}, two decision policies, associated each with a time scale, are needed. 
	While, in the fast time scale,  a  finite horizon training  policy is derived, in the slow time scale, an infinite horizon  position estimation  policy is required. The  combination of the latter will provide a solution to \eqref{problem:lower upper_level}.  In the  slow time scale, at each decision  epoch $n$,  the network  estimates the locations of a maximum of $U_{max}$ copilot groups and update the  large-scale fading coefficients accordingly. In the fast time scale,  between  two  upper level decision epochs ($n$ and $n+1$),   a  finite horizon training  policy is derived based on the  updated user locations that result from the upper level  optimization.

		\subsection{Fast time scale: learning an optimal  training strategy for finite horizon }\label{sec:copilot_grouping_algo}
		In this subsection,  we focus on solving the  lower level  planning problem in order to derive $R^{max}(x,\vec u)$, see equation~\eqref{eq:optimality_equation_max}.  
		We consider a deterministic sequential-decision making problem with reward given in equation~\eqref{Rlow}.
		The actions of the network on  the fast time scale are optimized while assuming a given belief state $x$, an initial  state $\vec d(0)=(0,\ldots,0)$ for all $n=0,1,\ldots$ and a given  action in the upper level $\vec u$.
		The  control horizon  $H$ is  selected to  be equal to the  large-scale  fading coherence block.  Without loss of generality, we consider $n=0$.
		The problem  of optimal users scheduling for uplink training  can be formulated as follows:
		
		\begin{align}\label{eq:finite_horizon_training_scheduling}
		\max_{\vec\phi^{low}_{x,\vec u}\in\overline\Phi^{low} }\{\sum_{t=t_{0}}^{t_{H-1}}\sum_{g=1}^{N_G}\sum_{l=1}^C\left(1-\frac{1}{T_s}\sum_{i=1}^{N_G} a_i(t)\right)\log\left(1+\hbox{SINR}_{gl}^{MRC}(\vec d(t), x,\vec u)\right)  \}, 
		\end{align}
		with
		\begin{align*}
		&\sum_{g=1}^{N_G} a_g(t)\leq \tau, \forall t=t_1,\ldots,t_{H-1}  \;\text{and} \;\;\vec d(0)=(0,\ldots,0).
		\end{align*}
		
		A naive approach to solve problem~\eqref{eq:finite_horizon_training_scheduling} is to generate all  $H$-length sequences of actions and then select the  sequence that  results in the  higher CASE after $H$ slots (brute force). Clearly, this approach can be quite computationally prohibitive when the action space and the optimization horizon are large. A more  appropriate approach is to use the Dynamic Programming (DP) algorithm, more precisely value iteration, see \cite{bertsekas1995dynamic} (based on the Bellman Equation). The DP approach can be used for sequential decision making problems like the one proposed in Eq.~\eqref{eq:finite_horizon_training_scheduling}. 
		\begin{remark}\label{eq:remark4}
			We note that solving (\ref{eq:finite_horizon_training_scheduling})     using the  DP approach can be computationally expensive for large optimization horizons $H$ with  a running time $\mathcal{O}((H-1) \left| X\right| \left| A\right|)$. 
			Consequently,  we provide an algorithm with lower complexity in order to  derive an approximate policy that reaches a guaranteed fraction of the optimal solution.
		\end{remark}
		As mentioned in Remark \ref{eq:remark4}, the DP approach results in  a long running time that can hinder the  uplink training  procedure. Consequently, we adopt an alternative approach  and  trait problem~\eqref{eq:finite_horizon_training_scheduling} by  combinatorial  optimization. Expressing  the  CSI  delays $\vec d (t_j)=(d_1(t_j),\ldots,d_{N_G}(t_j))$ as a function of the action  vectors  $\vec a(t)=(a_1(t),\ldots,a_{N_G}(t)), \; \forall  t=t_0,\ldots,t_{j-1}$ and $\vec d (t)=(d_1(t),\ldots,d_{N_G}(t)), \; \forall  t=t_0,\ldots,t_{j}$, is now in order. Recall  the definition  of the  deterministic fast time scale delay transition
		~\eqref{eq:deterministic_trans}. The  delay  $d_g(t_j), \forall g=1,\ldots,N_G$, can be  written as follows
		\begin{align}\label{eq:delay_function}
		d_g(t_j)= t_j \prod\limits_{t=t_1}^{t_j} (1-a_g(t))+ \sum\limits_{t=t_1}^{t_j}   t\; a_g(t_j-t)    \prod\limits_{h=t_j-t+1}^{t_j} (1-a_g(h)).
		\end{align}
		Consequently, the objective function in problem~\eqref{eq:finite_horizon_training_scheduling} can  be transformed into the following
		\begin{align}\label{eq:finite_combinatorial}
		\max_{\vec a(t_{0}),\ldots,\vec a(t_{H-1})}\{\sum_{t=t_{0}}^{t_{H-1}}\sum_{g=1}^{N_G}\sum_{l=1}^C\left(1-\frac{1}{T_s}\sum_{i=1}^{N_G} a_i(t)\right)\log\left(1+\hbox{SINR}_{gl}^{MRC}(\vec d(t), x,\vec u)\right)  \}, 
		\end{align}
		with
		\begin{align}
		&\sum_{g=1}^{N_G} a_g(t)\leq \tau, \forall t=t_1,\ldots,t_{H-1},\\
		&\hbox{SINR}_{gl}^{MRC}(\vec d(t_j), x,\vec u) =\frac{(M-1)(\beta_{gl}^{[l]})^2(\rho_{gl}^{[l]})^{2(t_j \prod\limits_{t=t_1}^{t_j} (1-a_g(t))+ \sum\limits_{t=t_1}^{t_j}   t\; a_g(t_j-t)    \prod\limits_{h=t_j-t+1}^{t_j} (1-a_g(h)))}}{(M-1)\times I_{gl}^p + I_{gl}^n},\nonumber
		\end{align}
		$I^p_{gl}$ and $I^n_{gl}$ are also defined accordingly by  combining Eq~\eqref{eq:interference} and Eq~\eqref{eq:delay_function}. The following Theorem helps to derive an efficient algorithm  to  solve problem~\eqref{eq:finite_combinatorial}.

		\begin{Theorem}
			Problem~\eqref{eq:finite_combinatorial} is equivalent to maximizing a submodular  set function subject  to  matroid  constraints. 
		\end{Theorem}
		\begin{IEEEproof}
			See appendix C.
		\end{IEEEproof}
		The  structure of  problem~\eqref{eq:finite_combinatorial} is quite convenient. In fact, even-though the  objective function  is not monotone, efficient approximation algorithms  exist for the non-monotone submodular  set function case. 
		In this  work, we  make use of the approximation algorithm proposed in \cite{sub_non_monotone} which provides a 
		$\left( \frac{1}{k+2+\frac{1}{k}+\epsilon}\right) $-approximation  of the optimal solution  under $k$ matroid constraints.
		In our case, we consider $H-1 $ matroid constraints. Each one is associated with a given optimization stage $t, t=t_1,\ldots,t_{H-1}$.  Consequently, the  proposed algorithm in this subsection provides a  	$\left( \frac{1}{H+1+\frac{1}{H-1}+\epsilon}\right) $-approximation  of the optimal cumulative average spectral efficiency with a running time $(N_G(H-1))^{\mathcal{O}(H-1)}$ \cite{sub_non_monotone}. The detailed algorithm  is given in table  $I$.
		We define the  ground set $G=\{v_{1t_1}, \ldots, v_{N_Gt_1},\ldots,v_{1t_{H-1}},\ldots, v_{N_Gt_{H-1}}\} $, where each element $v_{gt}$ represents the scheduling of copilot group $g$ for training at slot $t$. We also  define the sets $\mathcal{I}_{t},\; t=t_1,\ldots,t_{H-1} $. Each $\mathcal{I}_t $ contains the selected elements at stage $t$ with $\lvert\mathcal{I}_t\rvert \leq \tau$. 
		\begin{center}
			\begin{tabular}{ l  }
				\hline
				\hline
				$1. $ Set $G_0=G$:\\
				$2. $ for $t_1<h<t_{H-1}$:\\
				$3. $ Apply Approximate local search Procedure (table  II) on the ground set $G_h$ to obtain\\
				a solution $S_h\subset G_h$ corresponding to the  problem: $\;\;\hbox{max}_{S}(R^{up}(S,x,\vec u): S \subset G_h)$\\
				$4. $ set  $G_{h+1}=G_h \setminus S_h $\\
				$5. $ Return the best solution ($R^{max}(x,\vec u)=\hbox{max}_{S_1,\ldots,S_{H-1}}(R^{up}(S_h,x,\vec u))$).\\
				\hline
				\hline
			\end{tabular}
			\captionof{table}{\textbf{Algorithm for Approximate Finite horizon training strategy }}
		\end{center}
		\begin{center}
			\begin{tabular}{ l  }
				\hline
				\hline
				\emph{Input}: Ground set $X$ of elements\\
				$1. $ Set $v \longleftarrow \hbox{argmax}_{u \in X}(f(u))$ and  $S\longleftarrow \{v\}$\\
				$2. $ While one of the following local operations applies, update	$S$ accordingly\\
				\textbullet   Delete Operation  on $S$:\\
				If $e \in S$ such that   $f(S\setminus \{e\})    >  (1+\frac{\varsigma}{N_G^4}) f(S)$ then $S\longleftarrow S\setminus\{e\} $\\
				\textbullet   Exchange Operation  on $S$:\\
				If $d \in X \setminus S$ and $e_h \in S \cup \{\emptyset\}\; (\text{for}\; t_{1}<h<t_{H-1})$ are such  that   $ (S\setminus\{e_h\}) \cup  \{d\} \in \mathcal{I}_h $ \\
				for all  $h$ and  $f((S\setminus \{e_1,\ldots,e_{H-1}\} )\cup  \{d\}) >  (1+\frac{\varsigma}{N_G^4}) f(S) $,\\ then $S \longleftarrow  (S\setminus \{e_1,\ldots,e_{H-1}\} )\cup  \{d\} $.\\
				\hline
				\hline
			\end{tabular}
			\captionof{table}{\textbf{Approximate Local  search Procedure }}
		\end{center}

\subsection{Slow time scale:  adapting to  user mobility }

Once the  fast time scale planning problem is solved, we  tackle the  infinite horizon  positioning  problem of the slow time scale. Since we have chosen to decompose \eqref{problem:lower upper_level} into two levels, the  combination of the  policies, in the two time scales,  will provide  an infinite horizon policy  that  solves \eqref{problem:lower upper_level}.  The mobility of each copilot  group   $g$ is modeled by an $L$-state Markov chain.  The  positions of users in  a  each  copilot group $g$ remain the  same  for a given period which is equal  to  the  large scale  fading coefficients coherence block and  evolves according to the probability transition matrix $P_{g}$.

Solving the slow time scale control problem, directly, becomes intractable for  a large number of users and  possible positions, owing to   the resulting complexity  of belief-state monitoring   \cite{POMDPreview}. Nevertheless,  practical methods exist if  policy optimality is  abandoned  for the sake of  convergence  speed. We adopt  the  approximate approach in Nourbakhsh et al. \cite{Nourbakhsh}, which  solves  a POMDP  by  exploiting its underlying  Markov Decision  problem (MDP). This  is  done  by  ignoring the  agent's confusion (uncertainty about users locations)  and assuming  that it is in its most likely state (MLS). Replacing  a complicated POMDP  Problem by its underlying MDP enables to considerably reduce  complexity  since the  belief space is replaced by a more practical and  smaller state space.

We now discuss in more details how the  upper level policy is derived. Particularly, in our case, the  state of  the  underlying MDP, at a given  decision  epoch,  $s \in \mathcal{S}$ is  an $ N_G \times 1$ vector whose elements represent the  location  of  all copilot groups. That is, $s = (\ell_{1}, \ldots,\ell_{N_G})$. The most likely positions of users, for  each  decision  epoch $n=0,1,\ldots$,   are  obtained  as
\begin{align}\label{eq:most_likely}
& \{\ell^*_{1}(n),\ldots,\ell^*_{N_G}(n)\} = \hbox{argmax}_{\{\ell_{1}(n),\ldots,\ell_{N_G}(n)\} \in\{1,\ldots,L\}^{N_G} }( \prod\limits_{g=1}^{N_G} \vec b_{g\ell_{g}} (n)),
\end{align}
Recall that the  belief position at decision epoch $n$ depends on the  belief state transition given in~\eqref{eq:belief_transition}.
Using (\ref{eq:most_likely}),   the  agent's uncertainty  about user locations is removed and the upper level planning problem  is transformed to  a more practical  MDP. The  resulting MDP  is  solved using  value iteration \cite{finite_hor}. At each iteration, the CP  updates its belief-state (according to ~\eqref{eq:belief_transition})  and assumes that the users are in their most likely positions  (according to (\ref{eq:most_likely})). Then, a training policy is derived in  the  fast time scale, based on the  assumed positions. Deriving the latter can be  done using the algorithm  in table $I$. 
This  provides  the  upper level  reward which is equivalent  to  the  $H$-horizon lower level  reward $R^{max}(x,\vec u)=\max_{\vec\phi^{low}_{x,\vec u}\in\overline\Phi^{low} }\{R^{up}(\vec\phi^{low}_{x,\vec u},x,\vec u)\}$. The same procedure is repeated until  deriving the best position estimation decision for each  most likely state.
Although the derived policy  provides only an  approximate location estimation strategy, it  enables, nevertheless, to  solve  a problem otherwise  intractable in realistic scenarios.

\section{Numerical Results} 

In this section we provide some numerical results to validate  the analytical expression derived in section III and to demonstrate the performance of the proposed training/copilot
group scheduling scheme. We also  showcase the performance of the  proposed uplink training  learning procedures. We compare the   obtained results for the proposed schemes with  a   reference model where all scheduled users take part in uplink training. Consequently, the reference model is characterized by  $0$ CSI delay for all users and higher  training overhead which taken to be equal to the number  of scheduled users per cell.
We consider  $C = 7$ hexagonal cells, each of which has a radius of $1.5 \;\text{Km}$. The possible  positions of the mobile users are generated randomly in each cell with  minimum distance of  $10\; \text{m} $  to their serving BSs.
The movement  velocities and directions  are generated randomly for all users.  User speeds are drawn randomly from $[4 Km/h ; 80 Km/h]$. This interval  covers pedestrian  public transportation and urban car movement speeds. The angle separating the movement direction of the mobile devices and the directions of their incident waves are drawn from $[0, 2\pi]$. The path-loss exponent is considered to  be equal to  $3.5$. A  coherence slot of $T_s=200$ samples is assumed. We also  consider a coherence  time  of  $1 \; \text{ms}$. The  system operates over a bandwidth  of $200 \text{MHz}$ as considered in 5G systems \cite{5G}.  Once the copilot groups formed, we consider $L=5$ possible  position combinations for each  group.
The  transition probabilities matrices $P_{g}, g=1,\ldots,N_G$ are also generated randomly with  $\sum_{j=1}^{L}  p_{g}(i,j) =1, \forall  g=1,\ldots,N_G, \forall i=1,\ldots,L$.
%

\begin{figure}[!]
	\centering	
	\includegraphics[width=8.5cm,height=7.5cm]{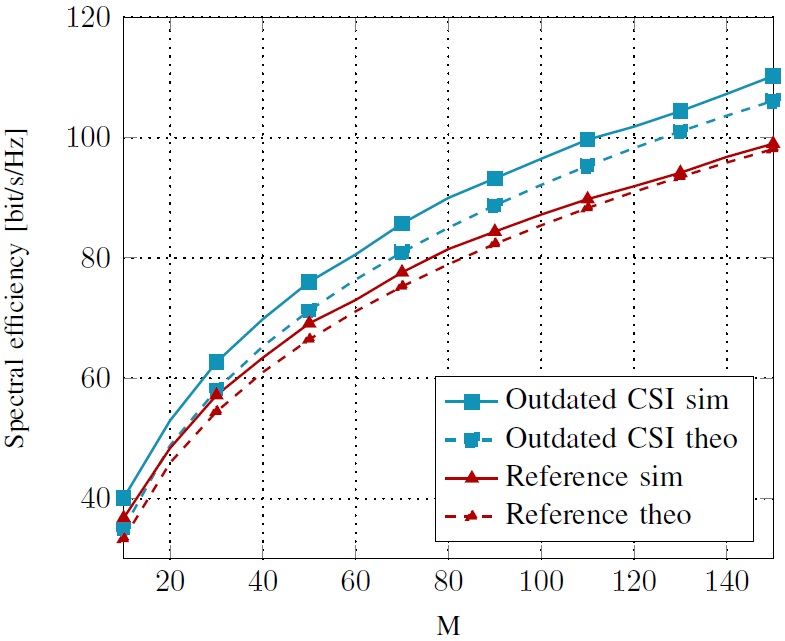}
	\captionsetup{font=footnotesize}
	\caption{Spectral efficiency for varying values of M }
\end{figure}

Figure 2 examines the tightness of the proposed analytical lower bound given in Theorem $1$. As can be  observed,  the  proposed  lower  bound  almost  overlap	with  the   simulation  curve. In addition, we readily see that  using  outdated CSI  with the implicated decrease of training  resources increases the SE by  $ 6.91 \;\text{bit/s/Hz}$  for  $M=50$. This gain  attains $ 11.2\; \text{bit/s/Hz}$ for $M=150$.

%
%
%

\begin{figure}[!]
	\centering	
	\includegraphics[width=8.5cm,height=7.5cm]{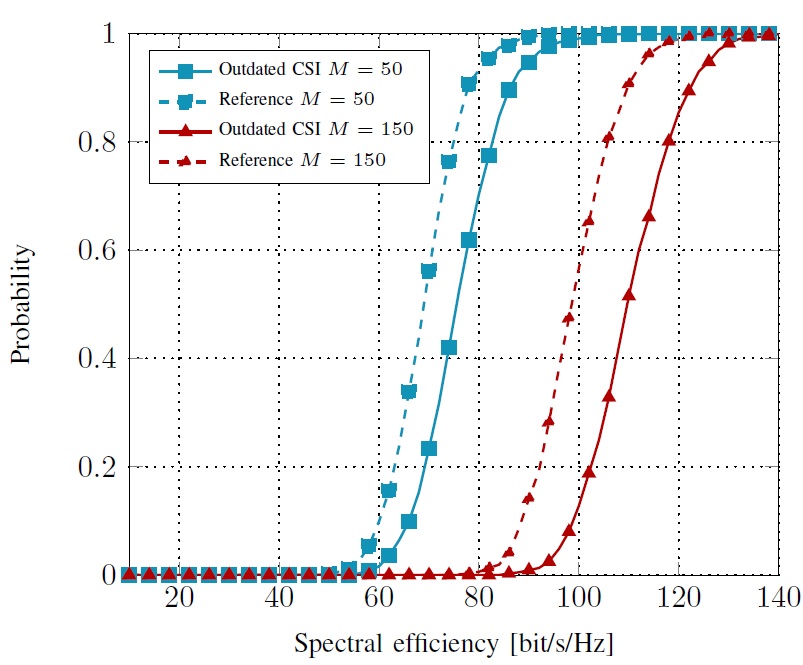}
	\captionsetup{font=footnotesize}
	\caption{Comparison of the  CDFs of spectral efficiency ($N_G=30$) }
\end{figure}

Figure 3 presents a comparison of  CDFs of the  achievable SE between the reference  model and  the proposed  training scheme for different numbers of antennas at the BS. 
For $50$ BS antennas, the  proposed training scheme achieves a gain in the $5\%$ outage  rate of  $ 6 \;\text{bit/s/Hz}$. For    $150$  antennas, the gain in the  $5\%$-outage  rate  grows to $ 8 \;\text{bit/s/Hz}$.  This increase in the performance is mainly due to the reduced  training  resources which can be  used to transmit more data.

%
%
%

\begin{figure}[!]
	\centering	
	\includegraphics[width=8.5cm,height=7.5cm]{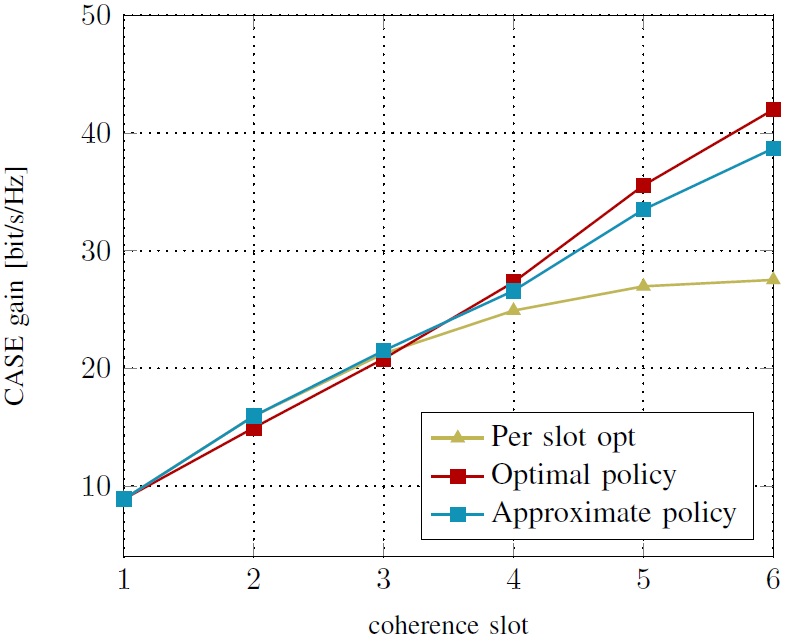}
	\captionsetup{font=footnotesize}
	\caption{ CASE gain for different  lower level algorithms}
\end{figure}

We now investigate the  perfomance of the proposed two time scale  learning algorithms proposed in sections IV.B,  IV.C  \& IV.D.	The performance is evaluated as the difference between the achievable CASE of  the  considered methods and a classical  Massive MIMO  TDD protocol. 

In	Figure 4, we illustrate the performance  of the uplink training learning algorithms in  sections IV.B \& IV.C.
The performance of optimal policy (Value iteration) and the approximate one (Algorithm table I) are compared with the  case where outdated CSI  is used  with  a per slot optimization.	The latter means that  the evolution of the correlation between the estimated CSI  and the actual  channel according to the time dimension is not  taken into  consideration  and  the scheduling  of copilot groups for uplink training  is optimized in order to  maximize the  ASE at each slot. Figure 4 shows that  using value iteration an the approximation algorithm in table I, the gain in CASE is maintained and attains $41.99\;\text{bit/s/Hz}$ and $38.7\;\text{bit/s/Hz} $ respectively, at the final  stage of the  optimization horizon $H$. However, although per slot optimization 
achieves also a gain  in CASE, we can  see that  this  method performs poorly in comparison with the proposed policies which shows the paramount importance of taking  the  time dimension into consideration when optimizing uplink training  decisions. Finally, due to its good performance,  we can deduce that  the approximate method (Algorithm table I) represents an efficient low complexity  substitute to the  more computationally prohibitive DP approach.

\begin{figure}[!]
	\centering	
	\includegraphics[width=8.5cm,height=7.5cm]{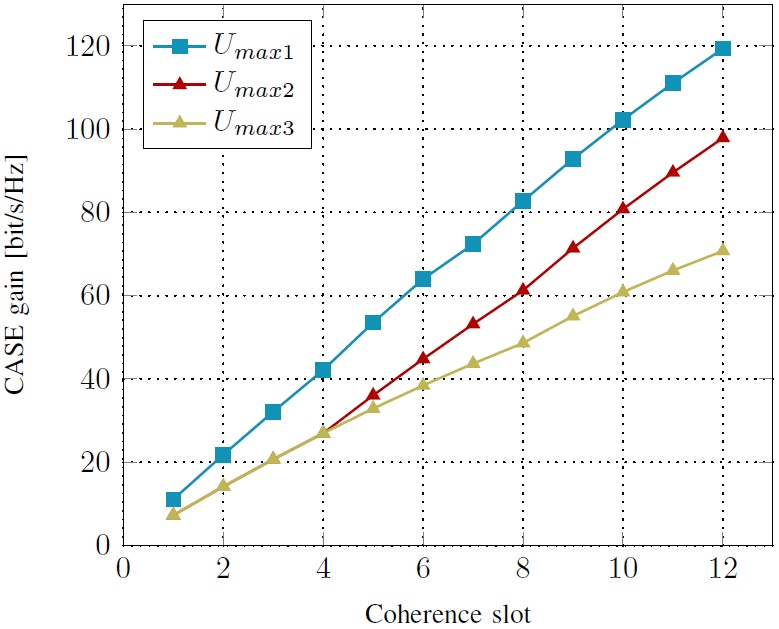}
	\captionsetup{font=footnotesize}
	\caption{CASE gain for different $U_{max}$ values ($H=4$, $n=0,\ldots,2$, $U_{max1}=N_G,\;U_{max2}=N_G-4  \;\text{and} \;U_{max3}=N_G-8$}
\end{figure}

%
%
%
%

In	Figure 5, we illustrate the  achievable  CASE  gain  after 3 upper level  decision epochs with $H=4$. In this  example $3$ values for $U_{max}$ were considered. As can be
readily  observed, decreasing  $U_{max}$ results in lower CASE  gain. This is quite intuitive since a lower $U_{max}$
results in more confusion about the  users locations. In fact,  the CPU  commits more errors when  inferring  users positions  from its belief states for lower $U_{max}$ values. Nevertheless, despite the positioning errors the proposed two time scale  learning approach is able to provide a considerable CASE gain of $110.26\;\text{bit/s/Hz} $, $97.863\;\text{bit/s/Hz} $ and  $70.73\;\text{bit/s/Hz} $ with $U_{max1},\;U_{max2}  \;\text{and} \;U_{max3}$  respectively, after $12$ $T_s$ slots.
These results did not showcase the  energy and signaling gains  that result from  reducing positioning estimation but are sufficient  to prove the advantages of allowing the network to proactively plan  its uplink training  decisions for long time periods.

\section{Conclusion}

In this paper, we analyzed the performance  of  an adaptive uplink training scheme for TDD massive MIMO systems, taking into account the  actual coherence time of the wireless channels, the  impact of  channel  aging and user mobility.
The idea is to  adapt the  periodicity of   CSI  estimation based on the actual coherence times.  We proposed a two time scale  control problem  in order to  allow the  network to  learn  the  best uplink  training policy taking into  consideration user mobility, channel coherence time and  practical signaling  overhead limitations.
In the fast time scale, the network  learns an optimal training policy  by choosing which  users are requested to send their pilot signals for  a predefined optimization horizon.	In  the slow time scale, owing to practical signaling and  processing overhead limitation, the  network is required to  choose which users are required to feedback their  positions, based on  their belief states. The  present  work shows that the aforementioned approach enables to leverage the  time evolution of  the  correlation between the wireless channel  and  the estimated CSI and  provides an impressive increase in the  achievable cumulative average spectral efficiency that cannot be obtained otherwise. Future work include the investigation  of similar procedures with fairness consideration  and  user traffic awareness.

\appendix


\section{Appendix}

\subsection{Proof of Theorem 1}\label{appendix:theorem1}
The network serves $N_g$ copilot groups, $\tau$  of which are scheduled for uplink training. At the reception, each  BS uses MRC receivers that are based on the latest  CSI estimates. BS  $l$ detects the signal of user $g$ in cell $l$ by applying the following filter $ u_{gl}(t)= \frac{\hat{g}^{[l]}_{ gl } (t-d_{g}) }{\lVert \hat{g}^{[l]}_{ gl } (t-d_{g})  \rVert}, t\geq d_{g}$,  where $\hat{g}^{[l]}_{ gl } (t-d_{g})$ denotes the latest available  CSI estimate  for user $g$ in cell $l$.
Consequently, the detected signal  of user $g$ in cell $l$ is given by the following
\begin{align}\label{TDD_appendix:eq:received_signal}
u_{gl}^\dagger(t) \frac{Y^{[l]}_u(t)}{\sqrt{P_u}}&= u_{gl}^\dagger(t)((\rho^{[l]}_{ gl })^{d_{g}} \hat{g}_{gl}^{[l]}(t-d_g)S_{gl}+\sum_{c \neq l}^C (\rho^{[l]}_{ gc })^{d_{g}}\hat{g}_{gc}^{[l]}(t-d_g)S_{gc}+\sum_{k\neq g}^{N_g}\sum_{c=1}^Cg_{kc}^{[l]}(t)S_{kc}\nonumber\\
&+\sum_{c=1}^C(\rho^{[l]}_{ gc })^{d_{g}}\tilde{g}_{gc}^{[l]}(t-d_g)S_{gc}+\sum_{c=1}^C\sum_{j=0}^{d_g-1} (\rho^{[l]}_{ gc })^{j} \sqrt{\beta_{gc}^{[l]}} \varepsilon_{gc}^{[l]}(t-j)S_{gc}+\frac{W_u(t)}{\sqrt{P_u}})\\
&=u_{il}^\dagger(t) (I_1(t)+I_2(t)+I_3(t))\nonumber,
\end{align}
with
\begin{align}
I_1(t) &= (\rho^{[l]}_{ gl })^{d_{g}}\hat{g}_{gl}^{[l]}(t-d_g)S_{gl},\\
I_2(t) &= \sum_{c \neq l}^C  (\rho^{[l]}_{ gc })^{d_{g}} \hat{g}_{gc}^{[l]}(t-d_g)S_{gc},\\
I_3(t) &=  \sum_{c=1}^C(\rho^{[l]}_{ gc })^{d_{g}}\tilde{g}_{gc}^{[l]}(t-d_g)S_{gc}+\sum_{c=1}^C\sum_{j=0}^{d_g-1} (\rho^{[l]}_{ gc })^{j} \sqrt{\beta_{gc}^{[l]}} \varepsilon_{gc}^{[l]}(t-j)S_{gc}+\sum_{k\neq g}^{N_g}\sum_{c=1}^Cg_{kc}^{[l]}(t)S_{kc}+\frac{W_u(t)}{\sqrt{P_u}}
\end{align}
Equation \ref{TDD_appendix:eq:received_signal} follows from the fact that $g_{kc}^{[l]}(t)=\sqrt{\beta_{kc}^{[l]}}h_{kc}^{[l]}(t)$, $h_{kc}^{[l]}(t)=\rho_{kc}^{[l]}h_{kc}^{[l]}(t-1)+\varepsilon_{kc}^{[l]}(t)$ for all $t$ and $g_{kc}^{[l]}(t)=\hat g_{kc}^{[l]}(t)+\tilde g_{kc}^{[l]}(t)$ for all $t$.
We note that $I_1(\cdot)$ refers to  the useful  signal,  $I_2(\cdot)$ represents the impact of pilot contamination and $I_3(\cdot)$ regroups the impact of  white noise, estimation error, non correlated interference due to users with different pilot sequences and the impact of aging. The instant SE attained by user $g$ in cell $l$ is:
\begin{align}\label{TDD_appendix:eq:SINR_user}
R_{g,l} = \left(1-\frac{\tau}{T}\right)\log\left(1+\frac{|u_{gl}^\dagger(t) I_1(t)|^2}{|u_{gl}^\dagger(t) I_2(t)|^2+|u_{gl}^\dagger(t) I_3(t)|^2}\right).
\end{align}
We now define $\overline R_{g,l}$ to be the average achievable sum rate of user $g$ in cell $l$, namely,
\begin{align}\label{TDD_appendix:eq:average_SINR_user}
\overline R_{g,l}=\mathbb{E}\left(\mathbb{E}\left(\left(1-\frac{\tau}{T}\right)\log\left(1+\frac{|u_{gl}^\dagger(t) I_1(t)|^2}{|u_{gl}^\dagger(t) I_2(t)|^2+|u_{gl}^\dagger(t) I_3(t)|^2}\right)\bigg|\hat g_{gl}^{[l]}(t-d_{g})\right)\right),
\end{align}
the last equality follows from the law of total expectation. Let us define $\overline R_{g,l}^0$ such that
\begin{align}\label{TDD_appendix:eq:average_knowing}
\overline R_{g,l}^0 = \mathbb{E}\left(\left(1-\frac{\tau}{T}\right)\log\left(1+\frac{|u_{gl}^\dagger(t) I_1(t)|^2}{|u_{gl}^\dagger(t) I_2(t)|^2+|u_{gl}^\dagger(t) I_3(t)|^2}\right)\bigg|\hat g_{gl}^{[l]}(t-d_{g})\right),
\end{align}
therefore, $\overline R_{g,l} = \mathbb{E}(\overline R_{g,l}^0 )$.
Based on the convexity of $\text{log}(1 +\frac{1}{x+a})$, and Jensen's  inequality  we obtain the following
\begin{align}\label{TDD_appendix:eq:average_SINR_zero}
\overline R^0_{g,l}\geq\left(1-\frac{\tau}{T}\right)\log\left(1+\frac{|u_{gl}^\dagger(t) (\rho_{gl}^{[l]})^{d_{g}}\hat g_{gl}^{[l]}(t-d_{g})|^2}{\mathbb{E}(|u_{gl}^\dagger(t) I_2(t)|^2|\hat g_{gl}^{[l]}(t-d_{g}))+\mathbb{E}(|u_{gl}^\dagger(t) I_3(t)|^2|\hat g_{gl}^{[l]}(t-d_{g}))}\right),
\end{align}
since 
\begin{align}
\mathbb{E}(|u_{gl}^\dagger(t)I_1(t)|^2|\hat g_{gl}^{[l]}(t-d_{g}))=|u_{gl}^\dagger(t)(\rho_{gl}^{[l]})^{d_{g}}\hat g_{gl}^{[l]}(t-d_{g})|^2,
\end{align}
We now aim at computing $\mathbb{E}(|u_{gl}^\dagger(t) I_j(t)|^2|\hat g_{gl}^{[l]}(t-d_{g})) \hbox{ for } j=2,3.$  In order to do so, we start by obtaining an alternative expression for $I_2(t)$, that is, 
\begin{align}
I_2(t) &= \sum_{c \neq l}^C\hat{g}_{gc}^{[l]}(t-d_g)S_{gc}= \hat{g}_{gl}^{[l]}(t-d_g) \sum_{c \neq l}^C \frac{\beta_{gc}^{[l]}}{\beta_{gl}^{[l]}}  S_{gc},	
\end{align}
since $\hat{g}_{gc}^{[l]}(t-d_g)=\hat{g}_{gl}^{[l]}(t-d_g)\frac{\beta_{gc}^{[l]}}{\beta_{gl}^{[l]}} $, $I_2(t)$ and $\hat{g}_{gl}^{[l]}(t-d_g)$ are correlated. Consequently, we obtain
\begin{align}\label{TDD_appendix:eq:average_I2}
&\E{|u_{gl}^\dagger(t) I_2(t)|^2|\hat g_{gl}^{[l]}(t-d_{g})} =   \left| u_{gl}^\dagger(t)  \hat{g}^{[l]}_{ gl } (t-d_{g})    \right|^2    \sum_{c \neq l}^C (\rho^{[l]}_{ gc })^{2 d_{g}} \frac{\beta^{[l]^2}_{gc}}{\beta^{[l]^2}_{gl}}.
\end{align}
We will now compute $\mathbb{E}(|u_{gl}^\dagger(t) I_3(t)|^2|\hat g_{gl}^{[l]}(t-d_{g}))$. First note that, $I_3(t)$ is independent of $\hat g_{gl}^{[l]}(t-d_{g})$ and since $u_{gl}^\dagger(t)$ has unit norm,  we obtain 
\begin{align}\label{TDD_appendix:eq:Ithree}
\mathbb{E}(|I_3(t)|^2 )&=\mathbb{E}\big( \sum_{k\neq g}^{N_g}\sum_{c=1}^C|g_{kc}^{[l]}(t)|^2+    \sum_{c=1}^C |(\rho_{gc}^{[l]})^{d_g}\tilde{g}_{gc}^{[l]}(t-d_g)|^2\\
&+\sum_{c=1}^C\sum_{j=0}^{d_g-1} |\sqrt{\beta_{gc}^{[l]}}(\rho_{gc}^{[l]})^j\varepsilon_{gc}^{[l]}(t-j)|^2    +|\frac{W_u(t)}{\sqrt{P_u}}|^2\big),\nonumber
\end{align}
where the  equality follows from noting the following four properties; (i) $S_{kc}\cdot S_{ic'}=0$ for all $k\neq i$ and all $c,c'\in\{0,\ldots,C\}$, (ii) $\mathbb{E}(Z W_u(t))=\mathbb{E}(Z)\mathbb{E}(W_u(t))=0$ for all random variables $Z$ that are independent of $W_u(t)$ (zero mean complex Gaussian noise), (iii) similar to the previous property, $\mathbb{E}(Z \varepsilon_{ic}^{[l]}(t))=\mathbb{E}(Z)\mathbb{E}(\varepsilon_{ic}^{[l]}(t))=0$ for all $Z$ independent of $\varepsilon_{ic}^{[l]}(t)$ (zero mean complex white Gaussian noise) and finally (iv) $g_{kc}^{[l]}$ and $\tilde g_{k'c'}^{[l]}$ are independent for all $(k,c)\neq(k',c')$.

We now compute the four terms in Equation \ref{TDD_appendix:eq:Ithree}. The last term, i.e., 
\begin{align}\label{TDD_appendix:eq:I4}
\mathbb{E}(|W_u(t)/\sqrt{P_u}|^2)= \frac{1}{P_u}. 
\end{align}

We now compute the third term in Equation \ref{TDD_appendix:eq:Ithree}, that is,
\begin{align}\label{TDD_appendix:eq:I3}
\mathbb{E}\left(\sum_{c=1}^C\sum_{j=0}^{d_g-1} |\sqrt{\beta_{gc}^{[l]}}(\rho_{gc}^{[l]})^j\varepsilon_{gc}^{[l]}(t-j)|^2\right)&=\sum_{c=1}^C\sum_{j=0}^{d_g-1} \beta_{gc}^{[l]}(\rho_{gc}^{[l]})^{2j}   (1-(\rho_{gc}^{[l]})^{2} )\\
&=\sum_{c=1}^C \beta_{gc}^{[l]} \frac{1-(\rho_{gc}^{[l]})^{2d_g}}{1-(\rho_{gc}^{[l]})^{2}}
(1-(\rho_{gc}^{[l]})^{2})\nonumber\\
\nonumber&=\sum_{c=1}^C \beta_{gc}^{[l]} ({1-(\rho_{gc}^{[l]})^{2d_g}}),\nonumber\nonumber
\end{align}
for the second equality we have used the expression of finite geometric sums since $(\rho_{gc}^{[l]})^2<1$ for all $g$ and $c$.
Next we compute the second term in Equation \ref{TDD_appendix:eq:Ithree}, namely,
\begin{align}\label{TDD_appendix:eq:I2}
&\mathbb{E}(\sum_{c=1}^C |(\rho_{gc}^{[l]})^{d_g}\tilde{g}_{gc}^{[l]}(t-d_g)|^2)= \sum_{c=1}^C (\rho_{gc}^{[l]})^{2d_g} \left( \beta_{gc}^{[l]}-\frac{(\beta_{gc}^{[l]})^2}{\frac{1}{P_p}+\sum_{b=1}^C\beta_{gb}^{[l]}}     \right).
\end{align}
the latter is satisfied due to the fact that the variance of $\tilde g_{gc}^{[l]}(t-d_{g})$ is given by $\beta_{gc}^{[l]}-\frac{(\beta_{gc}^{[l]})^2}{\frac{1}{P_p}+\sum_{b=1}^C\beta_{gb}^{[l]}}$ for all $g$ and $c$. We are left with the first term in Equation \ref{TDD_appendix:eq:Ithree}, that is,
\begin{align}\label{TDD_appendix:eq:I1}
&\sum_{k\neq g}^{N_g}\sum_{c=1}^C\mathbb{E}(|g_{kc}^{[l]}(t)|^2)=\sum_{k\neq g}^{N_g}\sum_{c=1}^C\mathbb{E}(|\sqrt{\beta_{kc}^{[l]}}     h_{kc}^{[l]}(t)|^2)=\sum_{k\neq g}^{N_g}\sum_{c=1}^C\beta_{kc}^{[l]},
\end{align}
Combining all four terms, that is, Equations \ref{TDD_appendix:eq:I4}, \ref{TDD_appendix:eq:I3}, \ref{TDD_appendix:eq:I2}, \ref{TDD_appendix:eq:I1} and \ref{TDD_appendix:eq:Ithree}, we obtain
\begin{align}\label{TDD_appendix:eq:Ilast}
\E{|u_{gl}^\dagger(t) I_3(t)|^2}&=\sum_{k\neq g}^{N_g}\sum_{c=1}^C \beta_{kn}^{[l]} +
\sum_{c=1}^C \beta_{gc}^{[l]} ({1-(\rho_{gc}^{[l]})^{2d_g}})
+\sum_{c=1}^C (\rho_{gc}^{[l]})^{2d_g} ( \beta_{gc}^{[l]}-\frac{(\beta_{gc}^{[l]})^2}{\frac{1}{P_p}+\sum_{b=1}^C\beta_{gb}^{[l]}} )\\
&=\sum_{k\neq g}^{N_g}\sum_{c=1}^C \beta_{kc}^{[l]} + \sum_{ c=1}^{C} ( \beta^{[l]}_{gc} -    \rho^{[l]^{2 d_g}}_{ gc }\frac{ \beta^{[l]^2}_{gc}}{\frac{1}{P_{p}}+\sum_{b=1}^{C}  \beta^{[l]}_{gb} })+\frac{1}{P_u}.\nonumber
\end{align}
Substituting the results in Equations \ref{TDD_appendix:eq:Ilast} and \ref{TDD_appendix:eq:average_I2} in Equation \ref{TDD_appendix:eq:average_SINR_zero}, we obtain
\begin{align}
\overline R_{g,l}^0& \geq \left(1-\frac{\tau}{T}\right)\log\left(1+\frac{(\rho_{gl}^{[l]})^{2d_{g}}|u_{gl}^\dagger(t) \hat g_{gl}^{[l]}(t-d_{g})|^2}{F}\right),
\end{align} 
with
\begin{align}\label{TDD_appendix:eq:average_F}
F=&|u_{gl}^\dagger(t) \hat g_{gl}^{[l]}(t-d_{g})|^2     \sum_{c \neq l}^C (\rho^{[l]}_{ gc })^{2 d_{g}} \frac{\beta^{[l]^2}_{gc}}{\beta^{[l]^2}_{gl}}+\sum_{k\neq g}^{N_g}\sum_{c=1}^C \beta_{kc}^{[l]} + \sum_{ c=1}^{C} ( \beta^{[l]}_{gc} -    \rho^{[l]^{2 d_g}}_{ gc }\frac{ \beta^{[l]^2}_{gc}}{\frac{1}{P_{p}}+\sum_{b=1}^{C}  \beta^{[l]}_{gb} })+\frac{1}{P_u}.
\end{align}
From Equation \ref{TDD_appendix:eq:average_SINR_user} and  \ref{TDD_appendix:eq:average_F} we obtain
\begin{align}\label{TDD_appendix:eq:average_line}
\overline R_{g,l}=\mathbb{E}(\overline R_{g,l}^0)=\mathbb{E}\left(\left(1-\frac{\tau}{T}\right)\log\left(1+\frac{(\rho_{gl}^{[l]})^{2d_{g}}}{G}\right)\right) \text{, where, }   G=\frac{F}{|u_{gl}^\dagger(t) \hat g_{gl}^{[l]}(t-d_{g})|^2}.
\end{align} 
Now, we apply Jensen's inequality to the right hand side (RHS) in Eq. \ref{TDD_appendix:eq:average_SINR_user}, that is,  
\begin{align}
\overline R_{g,l}&\geq\left(1-\frac{\tau}{T}\right)\log\left(1+\frac{(\rho_{g,l}^{[l]})^{2d_{g}}}{\mathbb{E}(G)}\right),
\end{align}
with
\begin{align}\label{TDD_appendix:eq:average_G}
\mathbb{E}(G)=& \sum_{c \neq l}^C (\rho^{[l]}_{ gc })^{2 d_{g}} \frac{\beta^{[l]^2}_{gc}}{\beta^{[l]^2}_{gl}}+ \mathbb{E}\left(\frac{1}{|u_{gl}^\dagger(t) \hat g_{gl}^{[l]}(t-d_{g})|^2 }\right)\\
&\cdot\left(\sum_{k\neq g}^{N_g}\sum_{c=1}^C \beta_{kc}^{[l]} + \sum_{ c=1}^{C} ( \beta^{[l]}_{gc} -    \rho^{[l]^{2 d_g}}_{ gc }\frac{ \beta^{[l]^2}_{gc}}{\frac{1}{P_{p}}+\sum_{b=1}^{C}  \beta^{[l]}_{gb} })+\frac{1}{P_u}\right).	\nonumber 
\end{align}
Note that $\left| u_{gl}^\dagger(t)  \hat{g}^{[l]}_{ gl } (t-d_{g})\right|^2$ has a Gamma distribution with parameters $(M,\frac{\beta^{[l]^2}_{gl}}{\frac{1}{P_p}+\sum_{b=1}^{C} \beta^{[l]}_{gb} })$. Consequently, the mean value of $\frac{1}{\left|u_{gl}^\dagger(t)  \hat{g}^{[l]}_{ gl } (t-d_{g})\right|^2}$ is equal to $\frac{1}{(M-1) \times \frac{\beta^{[l]^2}_{gl}}{\frac{1}{P_p}+\sum_{b=1}^{C} \beta^{[l]}_{gb} }}$. Combining this together with the results in Equations \ref{TDD_appendix:eq:average_line} and \ref{TDD_appendix:eq:average_G} we obtain the desired lower bound, that is,
\begin{align}
\overline R_{g,l} &  
\geq \left( 1-\frac{\tau}{T}\right)  \text{log} \left(1+ \frac{(M-1)(\beta^{[l]}_{gl})^2 (\rho^{[l]}_{ gl })^{2d_g} }{(M-1) I^p_{gl} +  I^n_{gl}  } \right),
\end{align}
where $I^p_{gl}$ and   $I^n_{gl}$ are given by
\begin{align}
& I^p_{gl} = \sum_{ c \neq l}^{C} (\rho^{[l]}_{ gc })^{2d_{g}} (\beta^{[l]}_{gc})^2, \hbox{ and } \\
& I^n_{gl}=  (\frac{1}{P_p}+\sum_{b=1}^{C} \beta^{[l]}_{gb}  )\cdot\left(\sum_{k\neq g}^{N_g}\sum_{c=1}^C \beta_{kc}^{[l]} + \sum_{ c=1}^{C} ( \beta^{[l]}_{gc} -    \rho^{[l]^{2 d_g}}_{ gc }\frac{ \beta^{[l]^2}_{gc}}{\frac{1}{P_{p}}+\sum_{b=1}^{C}  \beta^{[l]}_{gb} })+\frac{1}{P_u}\right).
\end{align}	
Summing the achievable SE of all grouped users concludes the proof.
\subsection{Proof of Theorem $2$}\label{appendix:theorem3}	
We consider the asymptotic regime where the number of  BS antennas $M$ grows  large. In this  case, the lower bound on  the SE of each user $g,l$ converges to the  following limit:
\begin{align}
\nonumber&  
\left( 1-\frac{\tau}{T_s}\right) \text{log} \left(1+ \frac{\beta^{[l]^2}_{gl} \rho^{[l]^{2 d_g}}_{ gl } }{\sum_{ b\neq l}^{C} \rho^{[l]^{2 d_g}}_{ gb } \beta^{[l]^2}_{gb} } \right).
\end{align}
The proposed framework is compared with  a reference  massive MIMO system  where, all scheduled users participate in uplink  training. The lower bound on the  achievable SE of each user $g,l$,  in the reference system, converges to the following:
\begin{align}
\left( 1-\frac{N_G}{T_s}\right) \text{log} \left(1+ \frac{\beta^{[l]^2}_{gl} }{\sum_{ b\neq l}^{C}  \beta^{[l]^2}_{gb} } \right).
\end{align}
The aim here, is to improve the  achievable SE of each scheduled users. Consequently, the  SE  of each user in the two considered systems should verify, $\; \forall\; g=1...N_G, l=1...C$:
\begin{align}
& \left( 1-\frac{\tau}{T_s}\right) \text{log} \left(1+ \frac{\beta^{[l]^2}_{gl} \rho^{[l]^{2 d_g}}_{ gl } }{\sum_{ b\neq l}^{C} \rho^{[l]^{2 d_g}}_{ gb } \beta^{[l]^2}_{gb} } \right) \geq  \left( 1-\frac{N_G}{T_s}\right) \text{log} \left(1+ \frac{\beta^{[l]^2}_{gl} }{\sum_{ b\neq l}^{C}  \beta^{[l]^2}_{gb} } \right),
\end{align}
which is equivalent  to the following condition:
\begin{align}
\nonumber&  \frac{\beta^{[l]^2}_{gl} \rho^{[l]^{2 d_g}}_{ gl } }{\sum_{ b\neq l}^{C} \rho^{[l]^{2 d_g}}_{gb } \beta^{[l]^2}_{gb} } \geq  \left(1+\frac{\beta^{[l]^2}_{gl} }{\sum_{ b\neq l}^{C}  \beta^{[l]^2}_{gb} }\right)^{\frac{T_s-N_G}{T_s-\tau}}-1.
\end{align}
We consider the extreme case where $\rho^2_{gl}=\bar{\rho}^{[{min}]^2}_{g}$ and $ \rho^2_{gb} =\bar{\rho}^{[{max}]^2}_{{g}}, \;\forall b\neq l$. 	Here $\bar{\rho}^{[{min}]}_{g}$ and $\bar{\rho}^{[{max}]}_{{g}}$ denote respectively the minimum and maximum channel autocorrelation  coefficients in  group $g$. This means that we assume the worst case scenario for each user.
Finally, by considering $SINR^{[\infty]}_{g,l}= \frac{\beta^{[l]^2}_{gl} }{\sum_{ b\neq l}^{C}  \beta^{[l]^2}_{gb} }$, we obtain $(14)$ which finishes the proof. 	
\subsection{Proof of Theorem $3$}\label{appendix:theorem4}
We start  by  demonstrating that the objective function  of  problem~\eqref{eq:finite_combinatorial}, is submodular. We note that the sum of  submodular functions is submodular. Consequently,  it is enough  to prove the  submodularity  of $f_g$ for a given copilot group $g$,  where $f_g$ is given by 
\begin{align}\label{eq:finite_horizon_training_scheduling_combinatorial}
f_g( \vec a(t_{0}),\ldots, \vec a(t_{H-1}), x,\vec u )=    \sum_{t=t_{0}}^{t_{H-1}}\sum_{l=1}^C  \big(1-\frac{1}{T_s}\sum_{i=1}^{N_G} a_i(t)\big)\log\big(1+\hbox{SINR}_{gl}^{MRC}(\vec d(t), x,\vec u)\big)  , 
\end{align}
We consider two  sets of action vectors, $\{\vec a(t) \in A , , t=t_{0},\ldots,t_{H-1} \}$ and $\{\vec a'(t) \in A , t=t_{0},\ldots,t_{H-1} \}$ such that, $\forall t=t_{0},\ldots,t_{H-1},\; \sum_{i=1}^{N_G} a_i(t) \leq \sum_{i=1}^{N_G} a'_i(t), \text{and}\; \forall i=1,\ldots, N_G,\; a_i(h)=1 \Rightarrow a'_i(h)=1$. These two sets of action vectors result, respectively,  in two sets of  delay vectors $\{\vec d (t), t=t_{0},\ldots,t_{H-1} \}$ and $\{\vec d' (t), t=t_{0},\ldots,t_{H-1} \}$ that can be obtained from $\vec a(t) $ and $\vec a'(t) $ according to (\ref{eq:delay_function}).
In order to prove the submodularity  of $f_g$, we need to prove that, for a given $h$ and $j$ such  that  $a_j(h)=a'_j(h) =0 $,     the marginal values of setting  $a_j(h)=1$ is higher than that of  $a'_j(h)=1 $ ,i.e  
\begin{align}
& f_g( \vec a(t_{0}),\ldots, \vec a(h) \oplus a_j(h),\ldots,\vec a(t_{H-1}), x,\vec u )-f_g( \vec a(t_{0}),\ldots, \vec a(h),\ldots,\vec a(t_{H-1}), x,\vec u )\geq\\
&	f_g( \vec a'(t_{0}),\ldots, \vec a'(h) \oplus a'_j(h),\vec a'(t_{H-1}), x,\vec u )-f_g( \vec a'(t_{0}),\ldots, \vec a'(h),\ldots, \vec a'(t_{H-1}), x,\vec u ).\nonumber
\end{align}	   
We will  distinguish between two cases, $j=g$ and $j\neq g$.  For the  first case, where $j=g$, the difference between the two marginal  values is given by:
\begin{align}
& \nonumber \Lambda - \Lambda'= \sum_{l=1}^C  
\log\big(1+\hbox{SINR}_{gl}^{MRC}(0, x,\vec u)\big)\big(\frac{\sum_{i=1}^{N_G} a'_i(h)-\sum_{i=1}^{N_G} a_i(h)}{T_s}\big)+ \big(1-\frac{\sum_{i=1}^{N_G} a'_i(h)}{T_s}\big)\\\nonumber
\nonumber	&\log\big(1+\hbox{SINR}_{gl}^{MRC}(d'_g(h), x,\vec u)\big) -\big(1-\frac{\sum_{i=1}^{N_G} a_i(h)}{T_s}\big)\log\big(1+\hbox{SINR}_{gl}^{MRC}(d_g(h), x,\vec u)\big)\\\nonumber
\nonumber	&+\sum_{t=h+1}^{t_{H-1}}\sum_{l=1}^C\big(1-\frac{\sum_{i=1}^{N_G} a_i(t)}{T_s}\big)\log\big(\frac{1+\hbox{SINR}_{gl}^{MRC}(d_g(t)-1, x,\vec u)}{1+\hbox{SINR}_{gl}^{MRC}( d_g(t), x,\vec u)}\big) -\big(1-\frac{\sum_{i=1}^{N_G} a'_i(t)}{T_s}\big)\\
&\log\big(\frac{1+\hbox{SINR}_{gl}^{MRC}(d'_g(t)-1, x,\vec u)}{1+\hbox{SINR}_{gl}^{MRC}( d'_g(t), x,\vec u)}\big)
\end{align}		
The difference  in  marginal values is positive as $\log\big(\frac{1+\hbox{SINR}_{gl}^{MRC}(d_g(t)-1, x,\vec u)}{1+\hbox{SINR}_{gl}^{MRC}( d_g(t), x,\vec u)}\big) $ is  decreasing as a function of  $d_g(t)$. We, now, consider  the case where $j \neq g$. In this case, we have
\begin{align}
& \Lambda - \Lambda'= \sum_{l=1}^C \frac{1}{T_s}\log\big(  \frac{1+\hbox{SINR}_{gl}^{MRC}(d'_g(h), x,\vec u)}{1+\hbox{SINR}_{gl}^{MRC}(d_g(h), x,\vec u)}\big)
\end{align}	
From the definition of $\vec a(h)$ and  $\vec a'(h)$, we have $\sum_{i=1}^{N_G} a_i(h) \leq \sum_{i=1}^{N_G} a'_i(h)$.\\ Hence $\hbox{SINR}_{gl}^{MRC}(d'_g(h), x,\vec u) \geq \hbox{SINR}_{gl}^{MRC}(d_g(h), x,\vec u)$ and the difference  in  marginal values  is also positive, in this case. 
Consequently, $f_g$ is submodular.  Concerning the matroid constraints, let us consider the ground set $G=\{v_{1t_1}, \ldots, v_{N_Gt_1},\ldots,v_{1t_{H-1}},\ldots, v_{N_Gt_{H-1}}\} $, where each element $v_{gt}$ represents the scheduling of copilot group $g$ for training at slot $t$.  It is clear that the constraints $(35)$  form a partition matroid on  $G$ \cite{matroid}. Consequently, problem~\eqref{eq:finite_combinatorial} is a  maximization of a submodular function  subject to matroid constraints.

\end{document}